\documentclass[10pt,journal]{IEEEtran}
\IEEEoverridecommandlockouts

\usepackage{amsmath}
 
\usepackage{amssymb}
\usepackage{mathptmx}

\usepackage{wrapfig} 
\usepackage{pifont}
\usepackage{mdframed}

\usepackage{bbding} 
\newcommand{\cmark}{\textcolor{green!80!black}{\ding{51}}}
\newcommand{\xmark}{\textcolor{red}{\ding{55}}}

\usepackage{url}
\usepackage{hyperref}
\hypersetup{
    colorlinks=true,
    linkcolor=blue,
    filecolor=magenta,      
    urlcolor=magenta,
    citecolor = magenta,
    pdftitle={An Example},
    pdfpagemode=FullScreen,
    }
\usepackage{tikz}
\usetikzlibrary{arrows.meta}
\usetikzlibrary{automata,positioning,backgrounds}
\usepackage{pgfplots}
\pgfplotsset{compat=1.18}

\usepackage{threeparttable}
\usepackage{capt-of}

\renewcommand*{\arraystretch}{1.5}%
\definecolor{tabred}{RGB}{230,36,0}%
\definecolor{tabgreen}{RGB}{0,116,21}%
\definecolor{taborange}{RGB}{250,124,30}%
\definecolor{tabbrown}{RGB}{171,70,0}%
\definecolor{tabyellow}{RGB}{251,253,169}%
\definecolor{warmone}{RGB}{247,248,242}%
\definecolor{warmtwo}{RGB}{248,240,234}%
\definecolor{warmthree}{RGB}{244,242,236}%
\definecolor{protocolteal}{RGB}{82,150,156}%
\newcommand*{\vcorr}{%
  \vadjust{\vspace{-\dp\csname @arstrutbox\endcsname}}%
  \global\let\vcorr\relax
}%

\usepackage{mathtools}

\usepackage{lscape}
\usepackage[figuresright]{rotating}
\usepackage[graphicx]{realboxes}
\usepackage{adjustbox}
\usepackage{caption}

\usepackage{subfigure}

\usepackage{colortbl} 
\usepackage{xfrac}

\usepackage{float}
\newcommand{\qw}[1]{#1}
\newenvironment{revision}{\begingroup}{\endgroup}

\usepackage{fbox} 
\usepackage{fancybox}

\usepackage{multirow}

\def\BibTeX{{\rm B\kern-.05em{\sc i\kern-.025em b}\kern-.08em
    T\kern-.1667em\lower.7ex\hbox{E}\kern-.125emX}}

\usepackage{hyperref}

\usepackage{tabularx,makecell, caption}
\usepackage{enumitem} 

\newcolumntype{L}{>{\arraybackslash}X}

\usepackage{multicol}
\usepackage{listings}
\usepackage{blindtext}
\usepackage{etoolbox,xstring,mfirstuc,textcase}
\usepackage{listings}

\usepackage{subfiles}
\usepackage[most]{tcolorbox}

\usepackage{xcolor}
\usepackage{amsthm}
\theoremstyle{plain}         
 
\newtheorem{thm}{Theorem}

 \newtheorem{lemma}{Lemma}
\newtheorem{defi}{Definition}

\newtheorem*{prf}{Proof}
\usepackage{siunitx}
\usepackage{comment}
\usepackage{indentfirst} 
\usepackage{framed} 

\usepackage[font=small,labelfont=bf,tableposition=top]{caption}
\usepackage{booktabs}
\usepackage{dashbox}
\usepackage{stfloats}   
\usepackage{placeins}   

\usepackage{listings}
\usepackage{color}
\usepackage{algorithm,algpseudocode} 

\begin{document}
\title{DAO to (Anonymous) DAO Transactions}

\author{
Minfeng Qi$^{\diamondsuit}$, Lin Zhong$^{\nabla}$, Qin Wang$^{\heartsuit}$
\\ \vspace{0.1in}
$^{\diamondsuit}$\textit{City University of Macau} $|$ $^{\nabla}$\textit{Binance} $|$  $^{\heartsuit}$\textit{CSIRO} 
}


\maketitle

\begin{abstract}
Blockchain assets are increasingly controlled by organizations rather than individuals. DAO treasuries, consortium wallets, and custodial exchanges rely on threshold authorization and multi-party key management, yet existing payment mechanisms still target single-user wallets, leaving no unified solution for organizational transfers. We formalize the problem of \emph{DAO-to-(anonymous)-DAO} transactions and present \textsc{Dao$^2$}, a framework that enables one threshold-controlled organization to pay another, optionally with recipient anonymity, while keeping received funds under distributed control. \textsc{Dao$^2$} combines three components: \emph{distributed key derivation} (DKD) for non-stealth child addresses, \emph{distributed stealth-address generation} (DSAG) for unlinkable one-time destinations, and \qw{\emph{threshold ECDSA}} for authorization. For ordinary transfers, the receiver derives a non-stealth address via DKD; for anonymous transfers, it derives a stealth address via DSAG. The sender then threshold-signs the payment, and the receiver redeems the funds without reconstructing any master secret. \qw{We formally prove its security and evaluate a prototype. The core of an anonymous DAO-to-DAO transaction for a typical-sized (e.g., 7-member) DAO finishes in under 28\,ms with about 2.8\,KB of communication, and scales linearly with DAO size.}
\end{abstract}

\begin{IEEEkeywords}
Blockchain, DAO, threshold signature, stealth address, distributed key derivation, privacy-preserving payments
\end{IEEEkeywords}

\section{Introduction}

Blockchain-based digital assets are increasingly managed by organizations rather than individual users. In practice, decentralized autonomous organizations (DAOs) \cite{wang2025understanding,tang2025decentralized}, custodians, consortium wallets, and project treasuries often operate under threshold authorization and multi-party key management rather than a single private key. This organizational shift raises a fundamental question:

\smallskip
\textit{How can one DAO transfer assets to another DAO while preserving recipient privacy and ensuring that the received funds remain under distributed control?}
\smallskip

Most existing payments are still designed for the \textit{single-user} setting. In Bitcoin and many mainstream public blockchains, transactions are publicly visible, and repeated use of addresses or related key material can reveal financial relationships among participants. Privacy-enhancing techniques such as stealth addresses, one-time addresses, and ring signatures help mitigate this leakage by preventing third parties from linking multiple incoming transactions to the same recipient.

These approaches usually place the scanning and spending keys with one user. A DAO instead distributes authorization and key material across several parties. Anonymous receipt must therefore keep both address generation and later spending under the DAO's threshold policy. In particular, converting a one-time address into a local private key  breaks that policy: the received output must already be represented by threshold shares, with no party learning the complete spending key.

A second challenge arises after funds are received. In UTXO-based systems such as Bitcoin, outputs are consumed once and fresh addresses or derived keys are routinely required for continued operation. Even in account-based systems, hierarchical derivation remains useful for wallet organization, audit separation, and policy-based asset management. BIP32-style hierarchical deterministic (HD) wallets solve this problem elegantly in the single-user setting, but they do not directly support a distributed environment in which no single party should reconstruct the master secret. Supporting realistic DAO-managed transactions therefore requires threshold signing, anonymous receiving, and \emph{distributed key derivation}, so that received assets remain inside a structured, threshold-controlled wallet.

Prior work studies threshold signatures~\cite{lindell2018fast,canetti2020uc,doerner2024threshold}, privacy-preserving payment addresses~\cite{vanSaberhagen2013cryptonote,wang2024dsag,pu2023stealth}, and hierarchical key derivation~\cite{bip32,zhong2023dkd} mostly as separate primitives. Current systems lack one transaction framework that lets a DAO receive through unlinkable one-time addresses, keep the funds under threshold control, and continue managing them within a derived wallet.

We address this gap by presenting a unified transaction framework, \textsc{Dao$^2$}, for organizational blockchain payments. Our framework supports two representative transaction types:
\begin{itemize}
    \item \textbf{DAO-to-DAO transactions}: the receiving DAO uses \emph{distributed key derivation} (DKD) to derive a non-stealth child address under threshold control, and the sending DAO threshold-signs a payment to that address. This mode applies to ordinary transfers on any blockchain (e.g., Bitcoin, Ethereum).
    \item \textbf{DAO-to-anonymous-DAO transactions}: the sender uses \emph{distributed stealth-address generation} (DSAG) with the receiver's threshold-held child key to produce an unlinkable one-time destination and then threshold-signs the payment. The receiver later detects the output, recovers one-time spending shares, and redeems the funds under threshold control. This mode applies when recipient privacy is required (e.g., Monero stealth payments; also achievable on Bitcoin/Ethereum via wallet-level support).
\end{itemize}
Among these, the second setting is the main technical focus because it requires recipient privacy, threshold-controlled redemption, and distributed derivation to work together within one coherent protocol.

The framework involves three core techniques. The first is \emph{distributed key derivation}, which adapts BIP32-style hierarchical derivation to a threshold setting: all DAO members locally update their key shares to derive a fresh child address without reconstructing the master secret. The resulting address is a standard (non-stealth) receiving address suitable for ordinary DAO-to-DAO transfers. The second is \emph{distributed stealth-address generation}, which further converts a threshold-held child public key into an unlinkable one-time destination through a distributed ECDH protocol. This produces a stealth address for anonymous receiving while preserving the ability to recover one-time spending shares under threshold control. The third is \emph{threshold signatures}, which serve as the authorization backbone for both outgoing payments and the redemption of received assets. The framework requires the threshold-signature shares to be linearly reconstructable discrete-log shares compatible with ECDH (as in Shamir-based threshold ECDSA); we evaluate the core with a 2-out-of-$n$ policy, motivated by real DAO deployments, where low-threshold settings such as 2-out-of-3 provide a practical balance between security and operational simplicity.

The modules connect through two additive updates. DKD adds the derivation offset \(\omega^{(k)}\) to every receiver share, and DSAG adds the stealth offset \(\rho^{(k)}\). Lagrange reconstruction then yields shares of \(b^{(k)}+\rho^{(k)}\), whose public key is the one-time destination \(D^{(k)}\). These shares enter threshold ECDSA directly, so neither the child secret nor the one-time spending key is reconstructed.

At a high level, the system operates as follows: (1)~each DAO runs a distributed key-generation or other compatible threshold-setup procedure to obtain private-key shares, public-key shares, and an aggregate public key; (2)~the receiver derives a child key, and the sender either pays that key directly or uses it to form a DSAG destination; (3)~the sender DAO threshold-signs the payment transaction and broadcasts it; (4)~blockchain consensus nodes confirm the transaction on chain; and (5)~the receiver detects the output, derives its one-time shares when DSAG is used, and later spends under threshold control. After successful redemption, the receiver keeps \((B^{(k)},cc^{(k)})\) as its next wallet state, allowing later receipts to continue along the same derivation lineage.

\textsc{Dao$^2$} combines anonymous receiving, distributed derivation, and threshold spending in one transaction flow for DAO-managed assets. If each side has one member, the threshold operations reduce to ordinary single-key operations and the same construction gives a user-level payment flow.

\smallskip
\noindent\textbf{Contributions.}
This paper makes the following contributions:
\begin{itemize}
    \item We formulate \emph{DAO-to-anonymous-DAO} transactions and present \textsc{Dao$^2$}, a unified payment framework (\S\ref{sec:system-model}--\S\ref{sec:framework}) that extends blockchain payments from single-user wallets to threshold-controlled organizations. Classical user-level payments arise as special cases when a DAO degenerates to a single user.

    \item We design an end-to-end protocol (\S\ref{sec:e2e-protocol}) with two core modules and a compatible authorization layer: \emph{distributed key derivation} for ordinary DAO-to-DAO receiving, \emph{distributed stealth-address generation} for anonymous receiving, and \emph{threshold signatures} for outgoing payments and redemption.

    \item We formally analyze security and evaluate a secp256k1 prototype (\S\ref{sec:security-analysis}, \S\ref{sec:implementation}, and Appendix~\ref{sec:security-proofs}). We prove transaction correctness, spending unforgeability, recipient indistinguishability, and Byzantine robustness. For a 7-member DAO, the core path takes under 28\,ms with about 2.8\,KB of counted public data.
\end{itemize}

\smallskip
\noindent\textbf{Relation to prior work.}
Threshold-signature systems provide distributed authorization but generally do not address anonymous receiving~\cite{lindell2018fast,canetti2020uc,zhong2023fast2n}. Stealth-address and anonymous-payment systems provide recipient privacy but typically assume single-user control~\cite{vanSaberhagen2013cryptonote,rivest2001ring,noether2016ringct,wang2024dsag}. HD wallets and distributed derivation support scalable key management but do not by themselves solve unlinkable organizational receiving~\cite{bip32,zhong2023dkd}. Our work targets the missing intersection of these three directions. Table~\ref{tab:comparison} summarizes the comparison.

\begin{table}[H]
\centering
\begin{threeparttable}
\caption{Comparison with prior work.}
\label{tab:comparison}
\footnotesize
\setlength{\tabcolsep}{3pt}
\renewcommand{\arraystretch}{1.15}
\begin{tabular}{rcccc}
\toprule
\multicolumn{1}{c}{\textbf{Approach}} & \textbf{TS} & \textbf{SA} & \textbf{DKD} & \textbf{E2E} \\
\midrule
Threshold ECDSA~\cite{lindell2018fast,doerner2024threshold} & \cmark & \xmark & \xmark & \xmark \\
\quad CryptoNote/RingCT~\cite{vanSaberhagen2013cryptonote,noether2016ringct} 
\quad  & \xmark & \cmark & \xmark & \xmark \\
DSAG~\cite{wang2024dsag} & \cmark & \cmark & \xmark & \xmark \\
BIP32 / Dist.\ KD~\cite{bip32,zhong2023dkd} & \xmark & \xmark & \cmark & \xmark \\
\cmidrule(rl){2-5}
\multicolumn{1}{c}{\textbf{\textsc{Dao$^2$}} (ours)} & \cmark & \cmark & \cmark & \cmark \\
\bottomrule
\end{tabular}
\begin{tablenotes}[flushleft]
\footnotesize
\item[] TS = threshold signing; SA = stealth address; DKD = distributed key derivation; E2E = end-to-end integration.
\end{tablenotes}
\end{threeparttable}
\end{table}

\smallskip
\noindent\textbf{Practical applications.}
Traditional blockchain payments focus on transfers between externally owned accounts or single-user wallets. Privacy-enhancing designs strengthen this model by hiding recipient identities. Our framework extends the model further to settings where both sender and receiver are governed by distributed decision-making and threshold-controlled cryptographic material. This makes it particularly suitable for DAO treasuries, consortium-controlled funds, custodial infrastructures, and other institutional digital-asset systems.

\section{Preliminaries}
\label{sec-preli}

We review the notation, stealth addresses, threshold signatures, and key derivation used in our construction.

\subsection{Basic Notation}

We work in an elliptic-curve group \( \mathbb{G} \) of prime order \( q \) with generator \( G \). Lowercase letters denote scalars in \( \mathbb{Z}_q \), and uppercase letters denote group elements. We write \( x \in_R \mathbb{Z}_q \) for uniform sampling, \( [n] = \{1,\ldots,n\} \), and \( \parallel \) for concatenation. Hash functions to \( \mathbb{Z}_q \) are denoted by \( H \), and BIP32-style derivation uses \(\text{HMAC-SHA512}\). For a qualified subset \( S \subseteq [n] \), let \( \lambda_{i,S} = \prod_{j \in S, j \neq i} \frac{j}{j-i} \) be the Lagrange coefficient associated with index \( i \in S \). When a scalar \( x \) is Shamir-shared among \( n \) parties with threshold \( t \), any subset \( S \) with \( |S| \ge t \) reconstructs \( x = \sum_{i \in S} \lambda_{i,S} x_i \). This linear relation is used repeatedly in our threshold-signing, distributed-stealth-address, and distributed-derivation components.

We use \(\mathsf{DLEQ.Prove}(x;G,X;P,Y)\) to prove that \(X=xG\) and \(Y=xP\) use the same scalar, and \(\mathsf{DLEQ.Verify}(G,X;P,Y;\pi)\) to verify the proof. Fiat--Shamir binds the proof to the session and the four group elements. We assume completeness, knowledge soundness, and zero knowledge in the random-oracle model.

\subsection{Stealth Addresses}

Stealth-address systems allow a recipient to publish long-term public information while receiving funds through unlinkable one-time addresses \cite{vanSaberhagen2013cryptonote,pu2023stealth}. In the standard construction, the recipient holds a scanning key pair \( (s, S=sG) \) and a spending key pair \( (s', S'=s'G) \). To send funds, the sender samples \( r \in_R \mathbb{Z}_q \), publishes \( R=rG \), and both parties derive the same Diffie-Hellman-style shared secret \( \sigma = H(rS) = H(sR) \). The one-time receiving address and its corresponding one-time secret key are then given by \( D = S' + \sigma G \) and \( d = s' + \sigma \pmod q \). The recipient scans candidate \( R \) values on chain, recomputes \( \sigma \), and tests whether the resulting \( D \) matches a transaction output. The main algebraic feature used later is that a stealth address is obtained by adding a hash-derived offset to the recipient's long-term spending key.

From the viewpoint of our framework, this primitive provides three relevant properties. First, it is \emph{correct}: the sender-side and receiver-side computations yield the same one-time key material. Second, it supports \emph{recipient unlinkability}, since external observers should not be able to associate the one-time address \( D \) with the recipient's long-term public key without learning the shared secret. Third, it preserves \emph{one-time spendability}: once the recipient recognizes the output, the corresponding secret key can be used exactly as an ordinary signing key for later redemption.

\subsection{Threshold Signatures}

A threshold signature splits a signing key among several parties but keeps one public verification key \cite{gennaro2007dkg,lindell2018fast,canetti2020uc,doerner2024threshold}. Party \(i\) holds a Shamir share \(x_i\) and public share \(X_i=x_iG\). Any qualified subset \(S\) satisfies \(x=\sum_{i\in S}\lambda_{i,S}x_i\) and \(PK=xG=\sum_{i\in S}\lambda_{i,S}X_i\). We use standard-output threshold ECDSA with an additive tweak: a public offset \(e\) changes each share to \(x_i+e\), each public share to \(X_i+eG\), and \(PK\) to \(PK'=PK+eG\). Public-key-dependent state is updated, and fresh preprocessing is created for \(PK'\). Key-independent setup may be reused, but nonces and presignatures are never reused across keys.

We require correctness, additive-tweak unforgeability, and robustness. Thus, a qualified subset can sign, while a sub-threshold set cannot sign under the base key or an honestly derived key.

\subsection{BIP32-Style Key Derivation}

BIP32 organizes wallet material as extended keys consisting of a secret or public key together with a chain code \cite{bip32,zhong2023dkd}. The non-hardened derivation relation is the one most relevant to our distributed construction. Given a parent extended key \( (x_{\mathrm{par}}, X_{\mathrm{par}}, cc_{\mathrm{par}}) \) and child index \( i \), one computes \( I = \text{HMAC-SHA512}(cc_{\mathrm{par}}, X_{\mathrm{par}} \parallel i) \), parses \( I = \omega \parallel cc_{\mathrm{child}} \), and derives \( x_{\mathrm{child}} = x_{\mathrm{par}} + \omega \pmod q \) together with \( X_{\mathrm{child}} = X_{\mathrm{par}} + \omega G \). Hence, child keys are obtained by adding a deterministically derived offset \( \omega \) to the parent key. Our distributed derivation protocol preserves exactly this additive parent-child relation while ensuring that the parent secret is never reconstructed by any party.

For the present work, the important properties are \emph{deterministic consistency}, namely that all honest parties derive the same offset and chain code from the same public derivation state; \emph{public-key consistency}, namely that private-key and public-key derivation remain algebraically aligned through the relation above; and \emph{hierarchical wallet continuity}, namely that freshly received assets can remain inside an organized derivation tree rather than being treated as isolated ad hoc keys.

\section{Threat Model \& Security Goals}

We define the participating parties, adversary, and security goals for DAO-to-DAO transfers.

\subsection{System Model}
\label{sec:system-model}

We consider two organizations, i.e., \( \mathsf{DAO}_1 \) and \( \mathsf{DAO}_2 \), with participant sets \( \{P^{(1)}_1,\ldots,P^{(1)}_{n_1}\} \) and \( \{P^{(2)}_1,\ldots,P^{(2)}_{n_2}\} \), respectively. System setup publishes public parameters \( \mathsf{pp} = (\mathbb{G}, q, G, H, \text{HMAC-SHA512}) \). The framework uses two protocol modules, distributed key derivation (\S\ref{sec-dkd1}) and distributed stealth-address generation (\S\ref{sec-dsag}), on top of a threshold-authorization primitive instantiated in \S\ref{sec-2N}. We write \( t_1 \) and \( t_2 \) for the sender-side and receiver-side authorization thresholds, respectively; our core experiments use \(t_1=t_2=2\).

At the sender side, a qualified subset \( S_1 \subseteq [n_1] \) jointly derives a stealth output for the receiver. At the receiver side, a qualified subset \( S_2 \subseteq [n_2] \) detects the output, reconstructs one-time signing capability in distributed form, and later authorizes spending. Each address-generation or recovery subset must meet its DAO's threshold.

Private state consists of secret shares, local randomness, and locally maintained derivation state. For anonymous transfers, we distinguish between the sender-visible session descriptor
\[
\delta^{(k)} = \bigl(B^{(k)}, cc^{(k)}, \mathsf{id}^{(k)}\bigr)
\]
and the public chain transcript
\[
\tau_{\mathrm{chain}}^{(k)} = \bigl(m_{\mathrm{pay}}^{(k)}, \sigma_{\mathrm{pay}}^{(k)}, D^{(k)}, \mathsf{id}^{(k)}, \xi^{(k)}\bigr).
\]
The descriptor \( \delta^{(k)} \) is the receiver-side coordination data needed by the designated sender-side subset to instantiate the transfer, whereas the privacy notion below is defined with respect to the publicly visible transcript \( \tau_{\mathrm{chain}}^{(k)} \). Off-chain coordination messages are not part of the public challenge view unless revealed through corruption.

\smallskip
\subsection{Adversary Model}
\label{sec-threat-model}

We consider a probabilistic polynomial-time adversary \(\mathcal{A}\) that chooses the corrupted participants in each DAO before setup. This set remains fixed throughout the execution. The adversary learns their full local states and controls all their messages.

We assume a \emph{sub-threshold adversary}: the fixed corrupted sets satisfy \(|C_1|<t_1\) and \(|C_2|<t_2\) over the entire key lineage, including all derived keys. For a 2-out-of-\(n\) policy, at most one member of each DAO is corrupted.

The adversary may delay messages, abort corrupted parties, send malformed shares, and propose inconsistent public values. We assume authenticated confidential channels for off-chain session data and reliable broadcast within each participating subset. Before a session starts, honest parties agree on its identifier, parent state, derivation tag, and participant sets. A shared chain-finality rule determines when a transaction is confirmed. Robustness concerns accepted outputs and visible aborts; it does not guarantee completion under arbitrary delays.

\smallskip
\noindent\textbf{Cryptographic assumptions.}
We assume DL, CDH, and DDH hardness in \(\mathbb{G}\). For the privacy proof, the stealth hash \(H\) and the ECDSA message hash \(H_{\mathrm{sig}}\) are independent random oracles; implementations should separate their input domains. DLEQ proofs are knowledge-sound and zero knowledge, optional commitments are binding, and threshold setup produces consistent shares. Spending security and signing robustness are conditional on the full interface requirement in \S\ref{sec-2N}, with fresh per-key preprocessing.
\subsection{Security Goals}
\label{sec:security-goals}

Against any PPT adversary \( \mathcal{A} \) in the model above, our framework is intended to satisfy the following security goals.

\begin{defi}[Transaction correctness]
The framework satisfies \emph{transaction correctness} if an honest execution with agreed inputs, qualified subsets, and delivered messages produces matching one-time keys, a valid redemption signature, and a consistent successor state. Writing \(D\) for the output and \(\{d_j\}_{j\in S_2}\) for the recovered shares, we require
\begin{multline*}
\Pr\!\Big[
D \neq \Big( \sum_{j \in S_2} \lambda_{j,S_2} d_j \Big) G
\ \lor\
\mathsf{TS.Verify}(D,m_{\mathrm{spend}},\sigma_{\mathrm{spend}})=0 \\
\quad \lor\ \mathsf{StateInconsistent}=1
\Big]
\leq \mathsf{negl}(\kappa),
\end{multline*}
where \( \mathsf{StateInconsistent} \) flags inconsistency in the post-transaction derivation state accepted by honest parties.
\end{defi}

\begin{defi}[Spending unforgeability]
The framework satisfies \emph{spending unforgeability} if a PPT adversary with the fixed sub-threshold corrupted sets above has negligible probability of winning the following experiment. It sees the public transcript, the corrupted parties' local views, and the results of authorized signing sessions, then outputs
\begin{itemize}
    \item a valid signature on a spending message not previously authorized by a qualified subset under the targeted honest key, or
    \item a secret \( x^* \) such that \( x^* G = PK \) for an honestly generated parent signing key \( PK \) or its derived one-time descendant.
\end{itemize}
Formally, if \( \mathsf{Exp}^{\mathrm{spend\text{-}uf}}_{\mathcal{A}}(1^\kappa)=1 \) denotes the event above, then
\[
\Pr\!\left[\mathsf{Exp}^{\mathrm{spend\text{-}uf}}_{\mathcal{A}}(1^\kappa)=1\right] \leq \mathsf{negl}(\kappa).
\]
\end{defi}

\begin{defi}[Recipient indistinguishability]
The framework satisfies \emph{recipient indistinguishability} if no PPT adversary can distinguish the intended receiver of the public chain transcript with more than negligible advantage. This experiment concerns an external observer: it receives no corrupted state, session DH value, or off-chain record identifying the selected receiver. The sender key is honestly generated and its secret remains hidden. Consider the experiment \( \mathsf{Exp}^{\mathrm{rec\text{-}ind}}_{\mathcal{A}}(1^\kappa) \):
\begin{enumerate}
    \item {\(\mathcal{A}\) selects two registered receivers. Their keys were generated honestly. The challenger gives \(\mathcal{A}\) the public parent states
    \( st_c^{\mathrm{pub}} = (B_c^{(k-1)}, cc_c^{(k-1)}) \) for \(c\in\{0,1\}\), but keeps the secret keys and shares. The two roots are generated independently, and the adversary cannot choose their key material or corrupt parties in this experiment.
    }
    \item {The challenger samples \( b \in \{0,1\} \) and fresh labels \( \mathsf{id}^{(k)} \) and \( \xi^{(k)} \). It derives the child state for receiver \(b\), runs the transfer with honest qualified subsets, and returns only
    \[
    \tau_{\mathrm{chain},b}^{(k)} =
    \bigl(m_{\mathrm{pay},b}^{(k)}, \sigma_{\mathrm{pay},b}^{(k)}, D_b^{(k)}, \mathsf{id}^{(k)}, \xi^{(k)}\bigr).
    \]
    The observer can derive both candidate descriptors from the public states and tags, but receives no off-chain record identifying which descriptor was selected. All other payment fields are fixed identically in both cases. The signature has the ordinary randomized-ECDSA distribution with an independent fresh nonce.
    }
    \item \( \mathcal{A} \) outputs a guess \( b' \).
\end{enumerate}
The framework is recipient-indistinguishable if
\[
\left|
\Pr[b'=b] - \frac{1}{2}
\right|
\leq \mathsf{negl}(\kappa).
\]
For unlinkability, compare two payments from the same honest sender to one challenge lineage with two payments to the two independent challenge lineages. Labels are freshly sampled, and all non-recipient payment fields are identical across the experiments. The observer receives only the corresponding chain transcripts and public root states. We require negligible distinguishing advantage, also over polynomially many such payments.
\end{defi}

\smallskip
\noindent\textit{Remark (scope of the privacy model).}
\qw{Definition~3 uses honestly generated receiver keys and honest qualified subsets on both sides.}
In particular, if the adversary has corrupted a sender-side participant who took part in the actual DSAG computation, that participant already knows the receiver's child public key \( B^{(k)} \) and the session descriptor \( \delta^{(k)} \), so receiver privacy cannot hold against such a party.
This is inherent to stealth-address-type constructions: a sender necessarily learns which receiver it is paying.
The guarantee is limited to external chain observers. An earlier participant who retained a descriptor and DH value is also excluded: on the same receiver lineage and under the same sender key, \(\Omega^{(k+1)}=\Omega^{(k)}+\omega^{(k+1)}A\) lets it recognize later destinations. Fresh labels do not prevent this use of a known DH value.

\begin{defi}[Byzantine robustness]
The framework satisfies \emph{Byzantine robustness} if, for every PPT adversary \( \mathcal{A} \) controlling a sub-threshold set of parties, the probability that honest parties accept malformed protocol data or end in an undetected inconsistent post-transaction state is negligible. Concretely, let \( \mathsf{Exp}^{\mathrm{rob}}_{\mathcal{A}}(1^\kappa)=1 \) if either
\begin{itemize}
    \item honest parties accept invalid shares or inconsistent public derivation data;
    \item two honest parties accept different successor states for the same agreed session and confirmed transaction.
\end{itemize}
A failed check or signing failure causes a visible abort; lack of enough responses causes no state commit. Neither event is itself a robustness violation. Then the framework is robust if
\[
\Pr\!\left[\mathsf{Exp}^{\mathrm{rob}}_{\mathcal{A}}(1^\kappa)=1\right] \leq \mathsf{negl}(\kappa).
\]
\end{defi}

\section{Construction}
\label{sec:framework}

We present the system flow and the three cryptographic components used by \textsc{Dao$^2$}.

\subsection{Overview}

Figure~\ref{fig:zk-dao} illustrates \textsc{Dao$^2$}, our framework for privacy-preserving transfers between threshold-controlled organizations. The sender side retains standard threshold authorization, while the receiver side combines distributed key derivation, stealth-address generation, and one-time threshold redemption so that incoming funds remain both unlinkable and organizationally controlled.


\begin{figure}[t!]
\centering
\includegraphics[width=0.99\linewidth]{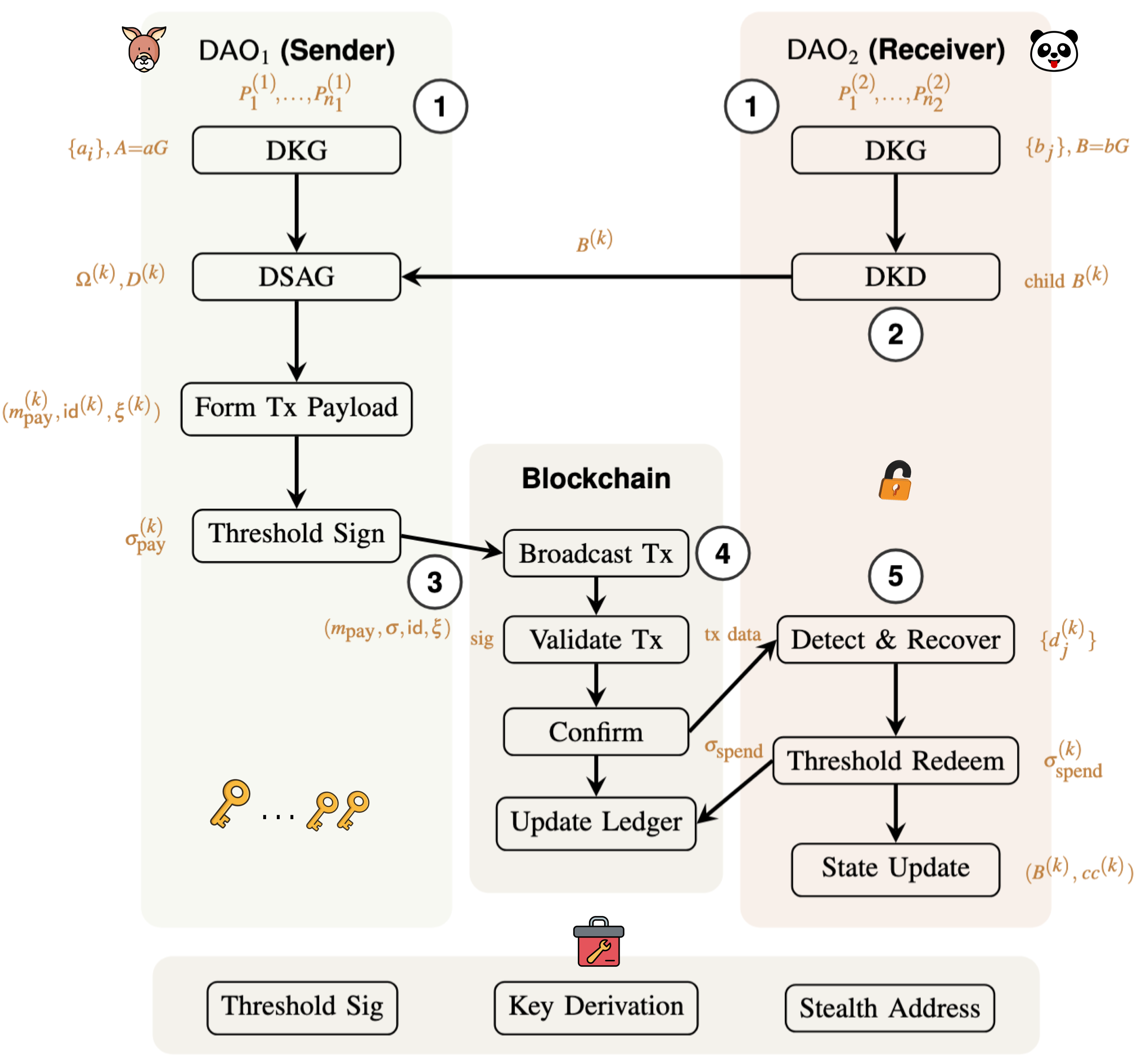}
\caption{\textsc{Dao$^2$} framework.
\ding{192}~DKG distributes key shares to each DAO.
\ding{193}~The receiver derives a child key (DKD); for anonymous
transfers the sender also generates a stealth destination (DSAG).
\ding{194}~The sender threshold-signs the payment.
\ding{195}~Blockchain validates and confirms the transaction.
\ding{196}~The receiver detects the output, recovers one-time shares,
redeems via threshold signing, and updates its state.}
\label{fig:zk-dao}
\end{figure}

At the construction level, \textsc{Dao$^2$} consists of two protocol modules and one underlying authorization primitive. The first module is distributed key derivation, which maintains a threshold-held extended-key state for the receiver DAO. The second is distributed stealth-address generation, which converts a derived receiver public key into an unlinkable one-time destination while preserving compatibility with later threshold spending. These two modules are used together with a compatible threshold-ECDSA interface, which serves as the authorization primitive for both outgoing payments and later redemption. \S\ref{sec:e2e-protocol} explains how these components are composed into one end-to-end transaction flow.

Concretely, \( \mathsf{DAO}_1 \) holds local threshold-signing states \(\{st_i^A\}_{i\in[n_1]}\), whose scalar components are shares \( \{a_i\}_{i \in [n_1]} \), together with public shares \(A_i=a_iG\) under public key \( A = aG \). The receiver \( \mathsf{DAO}_2 \) maintains a threshold-held derivation state
\[
\bigl(\{b_j^{(k)}\}_{j \in [n_2]},\; \{B_j^{(k)}\}_{j \in [n_2]},\; B^{(k)},\; cc^{(k)}\bigr),
\]
where \( B^{(k)} = b^{(k)}G \) is the current public key and \( cc^{(k)} \) is the current chain code. For each incoming transfer, the system advances from a parent state \( (B^{(k-1)}, cc^{(k-1)}) \) to a child state \( (B^{(k)}, cc^{(k)}) \), and the anonymous output is created relative to that child public key.

\subsection{Distributed Key Derivation}
\label{sec-dkd1}

The key-derivation module adapts BIP32-style non-hardened derivation to a threshold-controlled wallet setting \cite{bip32,zhong2023dkd}. Starting from a parent state
\[
\bigl(\{b_j^{(k-1)}\}_{j \in [n_2]},\; \{B_j^{(k-1)}\}_{j \in [n_2]},\; B^{(k-1)},\; cc^{(k-1)}\bigr),
\]
the parties derive a child state
\[
\bigl(\{b_j^{(k)}\},\; \{B_j^{(k)}\},\; B^{(k)},\; cc^{(k)}\bigr)
\]
for a fresh public derivation tag \( \mathsf{id}^{(k)} \), without reconstructing the parent secret \( b^{(k-1)} \).

For each derivation step \(k\), the parties compute
\[
(\omega^{(k)}, cc^{(k)}) = \text{HMAC-SHA512}\!\left(cc^{(k-1)}, B^{(k-1)} \parallel \mathsf{id}^{(k)}\right),
\]
interpret the first half as a scalar offset \( \omega^{(k)} \in \mathbb{Z}_q \), and perform the local additive update
\[
b_j^{(k)} = b_j^{(k-1)} + \omega^{(k)} \pmod q, \qquad
B_j^{(k)} = B_j^{(k-1)} + \omega^{(k)}G.
\]
Consequently, the aggregate child public key satisfies
\[
B^{(k)} = B^{(k-1)} + \omega^{(k)}G.
\]
This additive relation is the only DKD rule needed by the later transaction flow. It lets the sender side derive the child public key from the receiver-shared session descriptor, while only the receiver side can derive the matching child secret shares. If an implementation additionally wants proactive rerandomization after derivation, it may run any compatible share-refresh procedure \cite{gennaro2007dkg,kate2024nivss}; such maintenance is orthogonal to the core construction here.

\subsection{Distributed Stealth Address Generation}
\label{sec-dsag}

The DSAG module turns a receiver child public key into an unlinkable one-time destination that can still be redeemed under threshold control. Given a qualified sender-side subset \( S_1 \) holding shares \( \{a_i\}_{i \in S_1} \) of \( \mathsf{DAO}_1 \)'s signing secret and a receiver child public key \( B^{(k)} \), the module outputs a one-time public key \( D^{(k)} \) together with public metadata that enables \( \mathsf{DAO}_2 \) to recover the corresponding one-time secret shares in distributed form \cite{vanSaberhagen2013cryptonote,wang2024dsag,pu2023stealth}.

Each participating sender computes
\[
\Omega_i^{(k)} = a_i B^{(k)}.
\]
\qw{Each sender publishes this term with a DLEQ proof for \((G,A_i)\) and \((B^{(k)},\Omega_i^{(k)})\). Only verified terms are aggregated.}
By Lagrange aggregation, the distributed shared secret is
\[
\Omega^{(k)} = \sum_{i \in S_1}\lambda_{i,S_1}\Omega_i^{(k)} = aB^{(k)}.
\]
The senders then sample a fresh public label \( \xi^{(k)} \), define
\[
\rho^{(k)} = H(\Omega^{(k)} \parallel \xi^{(k)}), \qquad
D^{(k)} = B^{(k)} + \rho^{(k)}G,
\]
and publish the metadata \( (\mathsf{id}^{(k)}, \xi^{(k)}) \).

On the receiver side, a qualified subset \( S_2 \) reconstructs the same child state through \S\ref{sec-dkd1} and computes
\[
\Omega_j'^{(k)} = b_j^{(k)}A, \qquad
\Omega'^{(k)} = \sum_{j \in S_2}\lambda_{j,S_2}\Omega_j'^{(k)} = b^{(k)}A = \Omega^{(k)}.
\]
\qw{Each \(\Omega_j'^{(k)}\) is accepted only with a DLEQ proof for \((G,B_j^{(k)})\) and \((A,\Omega_j'^{(k)})\).}
Hence the receivers obtain the same offset \( \rho^{(k)} = H(\Omega'^{(k)} \parallel \xi^{(k)}) \) and derive one-time secret shares
\[
d_j^{(k)} = b_j^{(k)} + \rho^{(k)}, \qquad
D_j^{(k)} = B_j^{(k)} + \rho^{(k)}G.
\]
Because the same scalar \( \rho^{(k)} \) is added to every receiver share,
\[
\left(\sum_{j \in S_2}\lambda_{j,S_2} d_j^{(k)}\right)G = D^{(k)}.
\]
This is the only DSAG relation needed later. In particular, the module outputs ordinary Shamir-style shares \( \{d_j^{(k)}\} \) of the one-time spending secret, which are used through the compatible signing-state interface in \S\ref{sec-2N}.

\subsection{Threshold-ECDSA Interface}
\label{sec-2N}

\begin{revision}
The threshold-signature layer provides setup, public-key update, signing, and verification. Setup returns one local state per party and the aggregate public key:
\[
(\{st_i\}_{i\in[n]},X)\leftarrow\mathsf{TS.KeyGen}(1^\kappa,n,t),
\qquad X=xG.
\]
Here, \(st_i\) contains party \(i\)'s scalar share \(x_i\), public share \(X_i=x_iG\), and scheme-specific auxiliary state. For a public offset \(e\), each party applies the same additive update,
\[
x_i'=x_i+e,\qquad X_i'=X_i+eG,\qquad X'=X+eG,
\]
and obtains the corresponding local state as
\[
st_i'\leftarrow\mathsf{TS.Tweak}(st_i,e,X').
\]
A qualified subset \(S\) can then sign, and anyone can verify the result:
\[
\begin{aligned}
\sigma&\leftarrow\mathsf{TS.Sign}_S(m,\{st_i'\}_{i\in S}),\\
v&\leftarrow\mathsf{TS.Verify}(X',m,\sigma)\in\{0,1\}.
\end{aligned}
\]
Threshold ECDSA protocols provide the signing layer \cite{lindell2018fast,canetti2020uc,zhong2023fast2n,doerner2024threshold}. Our construction requires an adapter that updates all public-key-dependent auxiliary state when applying a tweak. Before signing, the parties create fresh preprocessing for \(X'\); nonces and presignatures are never reused across keys. Honest signatures must have the ordinary randomized-ECDSA output distribution.

\noindent\textit{Interface requirement.}
Against fixed sub-threshold corrupted sets, the adapted interface must remain correct and unforgeable across polynomially many public additive tweaks and authorized signing sessions, and must reject invalid contributions or abort visibly. Its security must cover the joint view of public shares, corrupted local states, signing transcripts, and the DH terms and DLEQ proofs disclosed to corrupted participants by DSAG. Public offsets and known session DH values need not remain secret. A forgery is a signature on a message not authorized by a qualified subset under the targeted base or derived key. This requirement is stronger than ordinary single-key ECDSA unforgeability.

Our spending-security and signing-robustness results are conditional on this interface requirement. The cited schemes motivate the signing layer; this paper does not implement or prove a concrete auxiliary-state adapter. The prototype checks standard-ECDSA output compatibility with threshold parameter \(t=2\).

\textsc{Dao$^2$} uses this primitive twice. First, a qualified subset of \(\mathsf{DAO}_1\) signs the payment under \(A\). Second, after the receiver recovers shares \(\{d_j^{(k)}\}\), each party applies \(\mathsf{TS.Tweak}\) with \(e=\rho^{(k)}\), and a qualified subset signs with fresh preprocessing for \(D^{(k)}\). On chain, \(D^{(k)}\) is an ordinary ECDSA key; the threshold protocol remains off chain.
\end{revision}

\section{End-to-End Protocol}
\label{sec:e2e-protocol}

We in this section describe how one transfer from \( \mathsf{DAO}_1 \) to \( \mathsf{DAO}_2 \) is executed. Figure~\ref{fig:dao2-protocol} gives a detailed protocol-level view of the anonymous transfer path.

For the \(k\)-th transfer, let \( \mathsf{id}^{(k)} \) denote a fresh public derivation tag, let \( S_1 \subseteq [n_1] \) and \( S_2 \subseteq [n_2] \) denote the qualified sender-side and receiver-side subsets participating in the session, and let \( T_1 \subseteq [n_1] \) and \( T_2 \subseteq S_2 \) denote the subsets used for outgoing authorization and redemption, with \( |T_1| \ge t_1 \) and \( |T_2| \ge t_2 \).

We assume throughout this section that the long-term threshold public key \(A\) of \( \mathsf{DAO}_1 \) and the current receiver-side parent state \( (B^{(k-1)}, cc^{(k-1)}) \) of \( \mathsf{DAO}_2 \) have already been established. 
In ordinary mode, the sender uses \(B^{(k)}\) directly as the receiving key: Step~1.3 still authorizes the payment, while the DSAG generation, detection, and one-time-share recovery steps are skipped. After confirmation, the receiver commits the child state and later signs under \(B^{(k)}\). Derivations on one receiver lineage are serialized: index \(k+1\) is allocated only after transaction \(k\) is committed or abandoned; parallel receipts use separate derivation branches.

The agreed session inputs and communication assumptions are those of \S\ref{sec-threat-model}. Parties bind every proof and signing request to that session and the same transaction record. A failed check leaves the persistent parent state unchanged. A successor state is committed only after the relevant transaction is confirmed under the shared finality rule; in anonymous mode, this is the redemption transaction.

\begin{figure*}[!t]
\centering
\tcbset{
  protocolstep/.style={
    enhanced,
    width=\linewidth,
    colback=white,
    frame hidden,
    boxrule=0pt,
    borderline={0.45pt}{0pt}{protocolteal,dashed},
    sharp corners,
    boxsep=0pt,
    left=5pt,
    right=5pt,
    top=3pt,
    bottom=3pt,
    before skip=2pt,
    after skip=2pt
  }
}
\newtcolorbox{protocolstepbox}{protocolstep}
\begin{mdframed}
\scriptsize
\setlength{\tabcolsep}{0pt}
\renewcommand{\arraystretch}{0.78}
\setlength{\abovedisplayskip}{0.5pt}
\setlength{\belowdisplayskip}{0.5pt}
\setlength{\abovedisplayshortskip}{0pt}
\setlength{\belowdisplayshortskip}{0pt}
\setlength{\parskip}{0pt}
\hspace*{0.018\textwidth}%
\begin{tabular}{@{}p{0.49\textwidth}@{\hspace{0.02\textwidth}}p{0.46\textwidth}@{}}
\noindent\textbf{\textsf{Phase I: Transaction Generation}} & \noindent\textbf{\textsf{Phase II: Recovery and Redemption}} \\

\smallskip
\noindent\underline{Step 1.1: Receiver-Side Child-Key Allocation}
The receiver derives the transaction-specific session descriptor from the current parent state \( (B^{(k-1)}, cc^{(k-1)}) \):
\smallskip

\begin{protocolstepbox}
\begin{enumerate}[nosep,leftmargin=*]
\renewcommand{\labelenumi}{(\alph{enumi})}
\item \textit{Derivation input:}
Use the current receiver-side parent state \( (B^{(k-1)}, cc^{(k-1)}) \) and a fresh tag \( \mathsf{id}^{(k)} \).
\item \textit{Offset derivation:}
\[
(\omega^{(k)}, cc^{(k)}) = \text{HMAC-SHA512}\!\left(cc^{(k-1)}, B^{(k-1)} \parallel \mathsf{id}^{(k)}\right)
\]
where the first 256 bits give the scalar offset $\omega^{(k)} \in \mathbb{Z}_q$ and the last 256 bits give the updated chaincode $cc^{(k)}$.
\item \textit{Child public key:}
\[
B^{(k)} = B^{(k-1)} + \omega^{(k)}G
\]
The additive structure ensures that each participant can locally derive the matching child secret share $b_j^{(k)} = b_j^{(k-1)} + \omega^{(k)}$ without reconstructing the parent secret.
\qw{Each receiver also sets \(st_j^{(k)}\leftarrow\mathsf{TS.Tweak}(st_j^{(k-1)},\omega^{(k)},B^{(k)})\). No preprocessing is created for this intermediate key.}
\item \textit{Session descriptor:}
Share \( (B^{(k)}, cc^{(k)}, \mathsf{id}^{(k)}) \) with the designated sender-side participants for the current transaction. This tuple is sufficient to instantiate the transfer and does not leak any receiver secret share.
\end{enumerate}
\end{protocolstepbox}
&

\smallskip
\noindent\underline{Step 2.1: Output Detection and Ownership Check.}
Upon observing a transaction carrying \( (D^{(k)}, \mathsf{id}^{(k)}, \xi^{(k)}) \), the receiver-side subset $S_2$ tests whether the output belongs to $\mathsf{DAO}_2$:
\smallskip

\begin{protocolstepbox}
\begin{enumerate}[nosep,leftmargin=*]
\renewcommand{\labelenumi}{(\alph{enumi})}
\item \begin{revision}\textit{Child state:}
Each participant derives \(\omega^{(k)}\) as in Step~1.1 and updates
\[
b_j^{(k)} = b_j^{(k-1)} + \omega^{(k)}, \qquad
B_j^{(k)} = B_j^{(k-1)} + \omega^{(k)}G
\]
\end{revision}
\item
\begin{revision}
\textit{Receiver-side shared secret:}
Each participant publishes \(\Omega_j'^{(k)}=b_j^{(k)}A\) with a DLEQ proof for \((G,B_j^{(k)})\) and \((A,\Omega_j'^{(k)})\). Only verified terms are aggregated:
\[
\Omega'^{(k)} = \sum_{j \in S_2}\lambda_{j,S_2}\Omega_j'^{(k)}
= b^{(k)}A = \Omega^{(k)}
\]
The last equality follows from ECDH.
\end{revision}
\item \textit{Candidate address test:}
\[
D'^{(k)} = B^{(k)} + H(\Omega'^{(k)} \parallel \xi^{(k)})G
\]
Accept the output only if \( D'^{(k)} = D^{(k)} \). If the test fails, the transaction is not addressed to this DAO and is silently skipped.
\item \textit{Detection record:}
On success, retain \( (\mathsf{id}^{(k)},\xi^{(k)},D^{(k)}) \) together with the verified \( (B^{(k)},cc^{(k)},\Omega'^{(k)}) \) as the pending redemption record passed to Step~2.2. The persistent wallet state is not advanced at this point.
\end{enumerate}
\end{protocolstepbox}
\\

\smallskip
\noindent\underline{Step 1.2: Distributed Stealth-Destination Generation.} The sender-side subset \( S_1 \) creates the one-time unlinkable destination for $\mathsf{DAO}_2$:
\smallskip

\begin{protocolstepbox}
\begin{enumerate}[nosep,leftmargin=*]
\renewcommand{\labelenumi}{(\alph{enumi})}
\item \textit{Local Diffie--Hellman terms:}
Each participating sender computes a partial DH term:
\[
\Omega_i^{(k)} = a_i B^{(k)} \qquad (i \in S_1)
\]
\item
\begin{revision}
\textit{Partial-DH consistency:}
Each sender publishes \(\Omega_i^{(k)}\) with a DLEQ proof for \((G,A_i)\) and \((B^{(k)},\Omega_i^{(k)})\). Only verified terms are aggregated. An optional commit-open round may be added before publication.
\end{revision}
\item \textit{Distributed secret aggregation:}
\[
\Omega^{(k)} = \sum_{i \in S_1}\lambda_{i,S_1} \Omega_i^{(k)} = aB^{(k)}
\]
\item \textit{Offset and one-time key:}
Sample a fresh public label $\xi^{(k)} \in_R \{0,1\}^{256}$ and compute:
\[
\rho^{(k)} = H(\Omega^{(k)} \parallel \xi^{(k)}), \qquad
D^{(k)} = B^{(k)} + \rho^{(k)}G
\]
\item \textit{Metadata formation:}
Attach \( (\mathsf{id}^{(k)}, \xi^{(k)}) \) to the payment transaction. $\mathsf{id}^{(k)}$ enables the receiver to locate the correct child key; $\xi^{(k)}$ enables reconstruction of $\rho^{(k)}$.
\end{enumerate}
\end{protocolstepbox}
&
\smallskip
\noindent\underline{Step 2.2: Distributed One-Time Share Recovery.}
\begin{revision}After detection, the receiver derives one-time shares without reconstructing the spending key:\end{revision}
\smallskip
\begin{protocolstepbox}
\begin{enumerate}[nosep,leftmargin=*]
\renewcommand{\labelenumi}{(\alph{enumi})}
\item \begin{revision}\textit{Offset:}\end{revision}
\[
\rho^{(k)} = H(\Omega'^{(k)} \parallel \xi^{(k)})
\]
\begin{revision}Since Step~2.1 establishes \(\Omega'^{(k)}=\Omega^{(k)}\), every honest receiver derives the same offset used by the sender.\end{revision}
\item \begin{revision}\textit{One-time shares:}\end{revision}
\[
d_j^{(k)} = b_j^{(k)} + \rho^{(k)}
\]
\begin{revision}Thus \(\sum_j\lambda_jd_j^{(k)}=b^{(k)}+\rho^{(k)}\), while no participant reconstructs the aggregate spending key.\end{revision}
\item \begin{revision}\textit{Public shares:}\end{revision}
\[
D_j^{(k)} = d_j^{(k)}G = B_j^{(k)} + \rho^{(k)}G
\]
\begin{revision}Each party publishes \(D_j^{(k)}\).\end{revision}
\item
\begin{revision}
\textit{Checks:}
Verify every \(D_j^{(k)}\stackrel{?}{=}B_j^{(k)}+\rho^{(k)}G\), then verify
\(D^{(k)}\stackrel{?}{=}\sum_{j\in S_2}\lambda_{j,S_2}D_j^{(k)}\). Reject any inconsistent share. Continue only if the verified shares still form a qualified subset; otherwise abort without changing the persistent state.
\end{revision}
\end{enumerate}
\end{protocolstepbox}
\\

\smallskip
\noindent\underline{Step 1.3: Outgoing Threshold Authorization.}
The sender-side transaction is finalized under the long-term threshold key \(A\):
\smallskip

\begin{protocolstepbox}
\begin{enumerate}[nosep,leftmargin=*]
\renewcommand{\labelenumi}{(\alph{enumi})}
\item \textit{Signer selection:}
Choose a qualified subset \( T_1 \subseteq [n_1] \) with $|T_1| \geq t_1$.
\item \textit{Signing input:}
Form the payment transaction \( m_{\mathrm{pay}}^{(k)} \) whose recipient field is the one-time key \( D^{(k)} \).
\item \textit{Threshold signing:}
Run the compatible threshold-ECDSA interface under public key \( A \):
\[
\sigma_{\mathrm{pay}}^{(k)} \!\leftarrow\! \mathsf{TS.Sign}_{T_1}\!\bigl(m_{\mathrm{pay}}^{(k)},\, \{st_i^A\}_{i \in T_1}\bigr)
\]
The signing protocol involves ephemeral nonce generation, partial-signature computation, and Lagrange-based aggregation (cf.\ \S\ref{sec-2N}).
\item \begin{revision}\textit{Freshness:}
Use fresh nonce material for every payment and never reuse a presignature.\end{revision}
\item \textit{Submission:}
Broadcast \( (m_{\mathrm{pay}}^{(k)}, \sigma_{\mathrm{pay}}^{(k)}, \mathsf{id}^{(k)}, \xi^{(k)}) \) to the blockchain network. The chain validates $\sigma_{\mathrm{pay}}^{(k)}$ against the public key $A$ without learning the threshold structure. After confirmation, erase \(\Omega_i^{(k)}\), \(\Omega^{(k)}\), and \(\rho^{(k)}\).
\end{enumerate}
\end{protocolstepbox}
&

\smallskip
\noindent\underline{Step 2.3: Threshold Redemption.} The receiver signs and updates its state:
\smallskip

\begin{protocolstepbox}
\begin{enumerate}[nosep,leftmargin=*]
\renewcommand{\labelenumi}{(\alph{enumi})}
\item \begin{revision}\textit{Signing input:}
Choose \(T_2\subseteq S_2\), \(|T_2|\geq t_2\), and form \(m_{\mathrm{spend}}^{(k)}\) under \(D^{(k)}\).\end{revision}
\item
\begin{revision}
\textit{Derived ECDSA state and redemption:}
Each party computes \(st_j^{D,(k)}\leftarrow\mathsf{TS.Tweak}(st_j^{(k)},\rho^{(k)},D^{(k)})\) and creates fresh preprocessing for \(D^{(k)}\). A qualified subset then signs:
\[
\sigma_{\mathrm{spend}}^{(k)} \!\leftarrow\! \mathsf{TS.Sign}_{T_2}\!\bigl(m_{\mathrm{spend}}^{(k)},\, \{st_j^{D,(k)}\}_{j \in T_2}\bigr).
\]
No nonce or presignature is reused across keys. On chain, \(D^{(k)}\) is an ordinary ECDSA key.
\end{revision}
\item \textit{Broadcast:}
Submit $(m_{\mathrm{spend}}^{(k)}, \sigma_{\mathrm{spend}}^{(k)})$ to the blockchain.
\item \textit{State update:}
After redemption is confirmed, store \((st_j^{(k)},B^{(k)},cc^{(k)})\) as the new state and mark \(\mathsf{id}^{(k)}\) as used.
\item \textit{Key erasure:}
Erase \(st_j^{D,(k)}\), its one-time share and nonce material, \(\rho^{(k)}\), \(\Omega_j'^{(k)}\), \(\Omega'^{(k)}\), and secret copies in the pending record.
\end{enumerate}
\end{protocolstepbox}
\end{tabular}
\end{mdframed}

\caption{End-to-end \textsc{Dao$^2$} transaction.}
\label{fig:dao2-protocol}
\vspace{-0.05in}
\end{figure*}

\subsection{Transaction-Generation Phase}

The first phase prepares and authorizes the payment.

\paragraph{\underline{Step 1.1}: Child-key allocation.}
Starting from the current receiver-side parent state \( (B^{(k-1)}, cc^{(k-1)}) \), the parties compute
\[
(\omega^{(k)}, cc^{(k)}) = \text{HMAC-SHA512}\!\left(cc^{(k-1)}, B^{(k-1)} \parallel \mathsf{id}^{(k)}\right)
\]
and the child public key
\[
B^{(k)} = B^{(k-1)} + \omega^{(k)}G.
\]
The tuple \( (B^{(k)}, cc^{(k)}, \mathsf{id}^{(k)}) \) is the sender-visible session descriptor for the current transaction. \qw{Each receiver sets \(st_j^{(k)}\leftarrow\mathsf{TS.Tweak}(st_j^{(k-1)},\omega^{(k)},B^{(k)})\). The parties create no preprocessing for this intermediate key.} The descriptor can be shared with the sender side without leaking the child secret shares.

\paragraph{\underline{Step 1.2}: Anonymous destination generation.}
Using the child public key \( B^{(k)} \), the sender-side subset \( S_1 \) runs the DSAG rule from \S\ref{sec-dsag}. Each participating sender computes \( \Omega_i^{(k)} = a_i B^{(k)} \). \qw{A DLEQ proof links this term to \(A_i=a_iG\). Only verified terms enter}
\[
\Omega^{(k)} = \sum_{i \in S_1}\lambda_{i,S_1}\Omega_i^{(k)} = aB^{(k)}.
\]
The senders then sample a fresh public label \( \xi^{(k)} \), derive
\[
\rho^{(k)} = H(\Omega^{(k)} \parallel \xi^{(k)}),
\]
and form the one-time destination
\[
D^{(k)} = B^{(k)} + \rho^{(k)}G.
\]
The public metadata attached to the transaction is \( (\mathsf{id}^{(k)}, \xi^{(k)}) \). At this point, \( D^{(k)} \) is already a valid one-time receiver key, but the payment still requires sender-side threshold authorization.

\paragraph{\underline{Step 1.3}: Outgoing authorization.}
A qualified subset \(T_1\) of \(\mathsf{DAO}_1\) uses its signing states \(\{st_i^A\}_{i\in T_1}\) and the threshold-signature scheme in \S\ref{sec-2N} to authorize payment to \(D^{(k)}\). This completes the generation phase. The chain verifies the payment signature, while the recipient field reveals only the one-time key \(D^{(k)}\). After confirmation, the senders erase their partial and aggregate DH values and the offset \(\rho^{(k)}\).

\subsection{Recovery-and-Redemption Phase}

The second phase begins when \( \mathsf{DAO}_2 \) observes a transaction carrying \( (D^{(k)}, \mathsf{id}^{(k)}, \xi^{(k)}) \).

\paragraph{Step 2.1: Output detection and ownership check.}
From \(\mathsf{id}^{(k)}\) and the current parent state, each receiver derives the same offset \(\omega^{(k)}\) and its child state:
\[
b_j^{(k)} = b_j^{(k-1)} + \omega^{(k)}, \qquad
B_j^{(k)} = B_j^{(k-1)} + \omega^{(k)}G.
\]
Each party then publishes \(\Omega_j'^{(k)}=b_j^{(k)}A\) with a DLEQ proof linked to \(B_j^{(k)}\). After verifying these proofs, the subset aggregates
\[
\Omega'^{(k)}=\sum_{j\in S_2}\lambda_{j,S_2}\Omega_j'^{(k)}
=b^{(k)}A=\Omega^{(k)}
\]
and tests the candidate address
\[
D'^{(k)}=B^{(k)}+H(\Omega'^{(k)}\parallel\xi^{(k)})G.
\]
The output is accepted only when \(D'^{(k)}=D^{(k)}\). A successful test creates a pending redemption record; it does not yet advance the persistent wallet state.

\paragraph{Step 2.2: One-time share recovery.}
For an accepted output, the receivers derive
\[
\rho^{(k)} = H(\Omega'^{(k)} \parallel \xi^{(k)})
\]
and compute one-time secret shares
\[
d_j^{(k)} = b_j^{(k)} + \rho^{(k)}, \qquad
D_j^{(k)} = d_j^{(k)}G.
\]
\qw{The receiver accepts only if every \(D_j^{(k)}=B_j^{(k)}+\rho^{(k)}G\) and}
\[
D^{(k)} = \sum_{j \in S_2}\lambda_{j,S_2} D_j^{(k)}
\]
holds. This is the point at which the anonymous output is mapped back into threshold-held secret state.

\paragraph{Step 2.3: Threshold redemption and state update.}
\qw{Each receiver applies \(\mathsf{TS.Tweak}\) with \(\rho^{(k)}\) and creates fresh preprocessing for \(D^{(k)}\). A qualified subset signs the redemption transaction. Preprocessing is never reused across keys.} After the redemption transaction is confirmed, the parties commit \((st_j^{(k)},B^{(k)},cc^{(k)})\) and mark \(\mathsf{id}^{(k)}\) as used. They erase the one-time ECDSA state, offset, receiver-side partial and aggregate DH values, and secret copies in the pending record.

\FloatBarrier
\section{Security Analysis}
\label{sec:security-analysis}

This section states the main security results for \textsc{Dao$^2$} and gives concise proof sketches. Full proofs and intermediate lemmas appear in Appendix~\ref{sec:security-proofs}.

We write \( a \) for the sender-side aggregate signing key, \( b^{(k)} \) for the receiver-side child key at epoch~\(k\), and adopt the notation of \S\ref{sec:e2e-protocol}. All negligible functions are in the security parameter~\(\kappa\).

\subsection{Transaction Correctness}
\label{subsec:sa-correctness}

\begin{thm}[Transaction correctness]
\label{thm:correctness}
For every honestly executed transaction involving qualified subsets \( S_1 \subseteq [n_1] \) and \( S_2 \subseteq [n_2] \), the one-time destination \( D^{(k)} \) produced in Phase~I satisfies
\[
D^{(k)} = \Big(\sum_{j \in S_2} \lambda_{j,S_2}\, d_j^{(k)}\Big)\, G,
\]
where \( d_j^{(k)} = b_j^{(k)} + \rho^{(k)} \) are the recovered one-time shares in Phase~II. Furthermore, any qualified receiver-side subset can produce a valid threshold signature under~\( D^{(k)} \), and all honest parties agree on the evolved derivation state
\( (B^{(k)},\, cc^{(k)}) \).
\end{thm}

\textit{Proof sketch.}
Adding the same public offset \(\omega^{(k)}\) to every receiver share preserves its Lagrange reconstruction because \(\sum_j\lambda_{j,S_2}=1\). Hence
\(\sum_j\lambda_{j,S_2}b_j^{(k)}=b^{(k)}\) and \(B^{(k)}=b^{(k)}G\) (Lemma~\ref{lem:dkd-consistency}). The sender and receiver then obtain the same shared secret:
\[
\sum_{i\in S_1}\lambda_{i,S_1}a_iB^{(k)}
=aB^{(k)}=b^{(k)}A
=\sum_{j\in S_2}\lambda_{j,S_2}b_j^{(k)}A.
\]
Therefore both sides derive the same \(\rho^{(k)}\). Adding this offset to the receiver shares gives
\(\sum_j\lambda_{j,S_2}d_j^{(k)}=b^{(k)}+\rho^{(k)}\), whose public key is exactly \(D^{(k)}\) (Lemmas~\ref{lem:shared-secret-consistency} and~\ref{lem:otk-reconstruction}). Fresh preprocessing for \(D^{(k)}\) then yields a valid threshold-ECDSA signature, while the deterministic derivation inputs make all honest parties adopt the same next state. \qed

\subsection{Spending Unforgeability}
\label{subsec:sa-threshold}

\begin{thm}[Spending unforgeability]
\label{thm:threshold-spending}
\begin{revision}
Assume knowledge-sound DLEQ proofs, the agreed-session model, and the full threshold-ECDSA interface requirement in \S\ref{sec-2N}. No PPT adversary \(\mathcal{A}\) with fixed sub-threshold corrupted sets can recover an honestly protected key or forge a valid spending transaction:
\[
\Pr\!\left[\mathsf{Exp}^{\mathrm{spend\text{-}uf}}_{\mathcal{A}}(1^\kappa)=1\right] \leq \mathsf{negl}(\kappa).
\]
\end{revision}
\end{thm}

\textit{Proof sketch.}
Condition on valid setup and accepted DH proofs. DKD and DSAG produce the public additive tweaks supported by the signing interface. The adversary may know these tweaks and a session's DH value; secrecy of \(\rho^{(k)}\) is not needed for spending security. Its joint view is covered by the interface requirement. An unauthorized spending signature is therefore an interface forgery. Recovery of a protected scalar key also gives such a forgery by signing a fresh unauthorized message (Lemma~\ref{lem:otk-secrecy}). Adding the negligible probability of an invalid accepted proof establishes the claim. \qed

\subsection{Recipient Indistinguishability}
\label{subsec:sa-privacy}

\begin{thm}[Recipient indistinguishability]
\label{thm:privacy}
Under DDH, independent random oracles \(H\) and \(H_{\mathrm{sig}}\), and honest randomized-ECDSA outputs, the framework satisfies recipient indistinguishability for the external observer in \(\mathsf{Exp}^{\mathrm{rec\text{-}ind}}_{\mathcal{A}}(1^\kappa)\) from \S\ref{sec:security-goals}:
\[
\left|\Pr[b'=b] - \tfrac{1}{2}\right| \leq \mathsf{negl}(\kappa).
\]
The same observer has negligible advantage in the multi-payment experiment of Definition~3.
\end{thm}

\textit{Proof sketch.}
Consider the public chain transcript
\[
\tau_{\mathrm{chain}}^{(k)} =
\bigl(m_{\mathrm{pay}}^{(k)}, \sigma_{\mathrm{pay}}^{(k)}, D^{(k)}, \mathsf{id}^{(k)}, \xi^{(k)}\bigr).
\]
Both candidate session descriptors are derivable from public data; the challenge view does not reveal which one was selected.
For a lineage rooted at \(B^{(0)}\), write \(B^{(k)}=B^{(0)}+\Delta_kG\). Replace its root DH value \(aB^{(0)}\) by a random point \(T\), and use \(T+\Delta_kA\) consistently at every descendant. DDH justifies this replacement. Fresh labels separate the stealth-hash inputs; unless a query hits a hidden DH input, each destination is uniformly masked.

The payment signatures require a separate step because they use the same secret \(a\). Lemma~\ref{lem:public-transcript-reduction} simulates randomized ECDSA from \(A\) by programming \(H_{\mathrm{sig}}\), with negligible probability of a prior-query conflict. This gives a reduction for the signed transcript. Applying the root replacement to each challenge lineage proves unlinkability without treating related DH values as independent (Lemma~\ref{lem:multi-tx-unlink}).
\qed

\subsection{Byzantine Robustness}
\label{subsec:sa-robustness}

\begin{thm}[Byzantine robustness]
\label{thm:robustness}
\begin{revision}
Assume consistent threshold setup, knowledge-sound DLEQ proofs, the agreed-session and channel model, and the signing-interface requirement in \S\ref{sec-2N}. No PPT static sub-threshold adversary can make honest parties accept invalid contributions or different successor states for the same confirmed session:
\[
\Pr\!\left[\mathsf{Exp}^{\mathrm{rob}}_{\mathcal{A}}(1^\kappa)=1\right] \leq \mathsf{negl}(\kappa).
\]
\end{revision}
\end{thm}

\textit{Proof sketch.}
Consistent setup and the agreed session inputs give matching derivation states (Lemmas~\ref{lem:dkg-robust} and~\ref{lem:derivation-consistency}). DLEQ proofs and per-share checks reject invalid DH and recovery contributions. Signing checks are supplied by the interface requirement. Parties commit only the successor of the agreed record after confirmation, so accepted successor states coincide. Failed checks cause a visible abort, and delayed or missing messages permit no state commit. This proves acceptance safety, not guaranteed completion. \qed

\subsection{State Exposure}
\label{subsec:sa-forward}
Anyone retaining an old extended public state and the later tags can recompute the intervening offsets and translate a current scalar share back to that member's earlier share. Erasing one-time states and DH records limits retained secrets, but additive derivation does not provide forward secrecy after a threshold compromise.


\begin{figure*}[!t]
\centering
\captionsetup{font=footnotesize,skip=2pt,singlelinecheck=true}

\begin{minipage}[t]{0.32\textwidth}
\vspace{0pt}\centering
\begin{tikzpicture}
\pgfplotsset{width=0.71\linewidth,height=2.38cm,scale only axis}
\begin{axis}[
    ybar,
    bar width=2.4pt,
    xlabel={DAO size $n$},
    ylabel={Time (ms)},
    symbolic x coords={3,5,7,10,15,20},
    xtick=data,
    ymin=0, ymax=42,
    tick label style={font=\scriptsize},
    label style={font=\scriptsize},
    legend image code/.code={\draw[#1] (0cm,-0.04cm) rectangle (0.10cm,0.04cm);},
    legend style={at={(0.03,0.97)},anchor=north west,
        font=\fontsize{4.7}{5.1}\selectfont,
        cells={anchor=west},legend columns=1,row sep=-1.2pt,
        inner xsep=1pt,inner ysep=1pt},
    grid=major,
    grid style={dashed, gray!30},
    every axis plot/.style={thick},
    enlarge x limits=0.12,
]
\addplot[fill=blue!60, draw=blue!80] coordinates {
    (3,1.23) (5,1.86) (7,2.33) (10,3.19) (15,4.83) (20,6.47)
};
\addplot[fill=red!60, draw=red!80] coordinates {
    (3,3.91) (5,6.62) (7,9.30) (10,13.90) (15,20.21) (20,27.84)
};
\addplot[fill=orange!60, draw=orange!80] coordinates {
    (3,4.79) (5,8.13) (7,11.20) (10,16.67) (15,24.67) (20,33.80)
};
\addplot[fill=green!60, draw=green!80] coordinates {
    (3,2.17) (5,2.15) (7,2.20) (10,2.14) (15,2.18) (20,2.24)
};
\legend{DKD,DSAG-Sender,DSAG-Receiver,ECDSA check}
\end{axis}
\end{tikzpicture}
\captionof{figure}{Core time versus DAO size~$n$ ($t=2$); excludes DLEQ and threshold ECDSA.}
\label{fig:comp-cost}
\end{minipage}
\hfill
\begin{minipage}[t]{0.32\textwidth}
\vspace{2.5pt}\centering
\begin{tikzpicture}
\pgfplotsset{width=0.71\linewidth,height=2.38cm,scale only axis}
\begin{axis}[
    xlabel={DAO size $n$},
    ylabel={Latency (ms)},
    xtick={3,5,7,10,15,20},
    ymin=0, ymax=85,
    tick label style={font=\scriptsize},
    label style={font=\scriptsize},
    legend image post style={xscale=0.65},
    legend style={at={(0.03,0.97)},anchor=north west,
        font=\fontsize{4.7}{5.1}\selectfont,
        cells={anchor=west},row sep=-1.2pt,
        inner xsep=1pt,inner ysep=1pt},
    grid=major,
    grid style={dashed, gray!30},
    mark size=1.8pt,
    every axis plot/.style={thick},
]
\addplot[color=blue, mark=triangle*, mark options={fill=blue}] coordinates {
    (3,7.45) (5,10.66) (7,14.13) (10,19.17) (15,27.86) (20,36.84)
};
\addplot[color=red, mark=square*, mark options={fill=red}] coordinates {
    (3,7.13) (5,10.25) (7,13.60) (10,19.09) (15,27.28) (20,36.12)
};
\addplot[color=black!70, mark=*, mark options={fill=black!70}, dashed] coordinates {
    (3,14.53) (5,20.84) (7,27.72) (10,38.39) (15,54.92) (20,73.20)
};
\legend{Phase I,Phase II,Total}
\end{axis}
\end{tikzpicture}
\captionof{figure}{Core latency versus DAO size~$n$ ($t=2$); DLEQ and threshold ECDSA are excluded.}
\label{fig:e2e-latency}
\end{minipage}
\hfill
\begin{minipage}[t]{0.32\textwidth}
\vspace{0pt}\centering
\begin{tikzpicture}
\pgfplotsset{width=0.71\linewidth,height=2.38cm,scale only axis}
\begin{axis}[
    ybar stacked,
    bar width=7pt,
    xlabel={DAO size $n$},
    ylabel={Bytes},
    symbolic x coords={3,5,7,10,15,20},
    xtick=data,
    ymin=0, ymax=8200,
    tick label style={font=\scriptsize},
    label style={font=\scriptsize},
    scaled y ticks=false,
    legend image code/.code={\draw[#1] (0cm,-0.04cm) rectangle (0.10cm,0.04cm);},
    legend style={at={(0.03,0.97)},anchor=north west,
        font=\fontsize{4.7}{5.1}\selectfont,
        cells={anchor=west},legend columns=1,row sep=-1.2pt,
        inner xsep=1pt,inner ysep=1pt},
    grid=major,
    grid style={dashed, gray!30},
    enlarge x limits=0.12,
]
\addplot[fill=blue!50, draw=blue!70] coordinates {
    (3,81) (5,81) (7,81) (10,81) (15,81) (20,81)
};
\addplot[fill=red!50, draw=red!70] coordinates {
    (3,633) (5,1023) (7,1413) (10,1998) (15,2973) (20,3948)
};
\addplot[fill=green!50, draw=green!70] coordinates {
    (3,128) (5,128) (7,128) (10,128) (15,128) (20,128)
};
\addplot[fill=orange!50, draw=orange!70] coordinates {
    (3,492) (5,820) (7,1148) (10,1640) (15,2460) (20,3280)
};
\legend{DKD,DSAG-Sender,Signatures,DSAG-Receiver}
\end{axis}
\end{tikzpicture}
\captionof{figure}{Communication with DLEQ and commit--open; excludes threshold ECDSA.}
\label{fig:comm-overhead}
\end{minipage}

\end{figure*}

\section{Implementation and Evaluation}
\label{sec:implementation}

\begin{revision}
We implemented DKD, DSAG, and a standard ECDSA compatibility check in Python with the \texttt{ecdsa} library on secp256k1. The test harness reconstructs the key only to check the final ECDSA output; it is not a threshold-ECDSA implementation. This check uses deterministic ECDSA; Theorem~\ref{thm:privacy} instead specifies randomized-ECDSA outputs for the protocol. Computation excludes DLEQ, DKG, per-key preprocessing, and interactive threshold signing. Communication includes DLEQ, optional sender commit-open, and the two final signatures, but excludes internal threshold-ECDSA messages. We ran each DAO-size setting 10 times on an Apple Silicon macOS machine and report the median. For the depth experiment, we advance the derivation chain once to each target depth and report the median of up to the final five per-step measurements.
\end{revision}

We evaluate five aspects: (i)~per-module computation cost as a function of DAO size~$n$; (ii)~end-to-end transaction latency across both protocol phases; (iii)~communication overhead; (iv)~comparison with baseline schemes; and (v)~key-derivation scalability with respect to derivation depth. Throughout, we fix the signing threshold at $t=2$ and vary $n \in \{3, 5, 7, 10, 15, 20\}$ to cover the practically relevant range for DAO treasuries.

\subsection{Per-Module Computation Cost (Fig.~\ref{fig:comp-cost})}
\label{subsec:eval-module}

To identify the main computational bottleneck, we time DKD, sender-side DSAG, receiver-side DSAG, and the standard-ECDSA compatibility check separately, with $n_1=n_2=n$ and $t=2$. At $n{=}7$, the four components take 2.33, 9.30, 11.20, and 2.20\,ms, respectively. Across the full range, DKD grows from 1.2 to 6.5\,ms, while the sender and receiver DSAG paths grow from 3.9 to 27.8\,ms and from 4.8 to 33.8\,ms. Thus, DSAG dominates the core computation and scales near-linearly with DAO size, whereas the ECDSA check remains nearly constant. These measurements exclude DLEQ costs.

\subsection{End-to-End Transaction Latency (Fig.~\ref{fig:e2e-latency})}
\label{subsec:eval-e2e}

We form Phase~I from DKD, DSAG-Sender, and the standard-ECDSA compatibility check, and Phase~II from DSAG-Receiver and the same check. We report the phases separately because transaction generation and redemption are executed by different DAO subsets and may occur at different times; the end-to-end value is their sum.

Our results show closely matched sender and receiver paths. At $n{=}7$, Phase~I takes 14.13\,ms and Phase~II takes 13.60\,ms, a difference of only 0.53\,ms, for a total of 27.72\,ms. The total rises from 14.53\,ms at $n{=}3$ to 73.20\,ms at $n{=}20$ and remains below 28\,ms for $n\leq7$. This balance indicates that neither side becomes a separate bottleneck as the DAO grows. The values measure the core algebra and ECDSA compatibility, not a full malicious-secure threshold-ECDSA execution.

\subsection{Communication Overhead (Fig.~\ref{fig:comm-overhead})}
\label{subsec:eval-comm}

We count all public data transmitted for one transaction using compressed EC points (33 bytes), 32-byte scalars and hash outputs, and 64-byte ECDSA signatures. The accounting includes DLEQ proofs, the optional sender commit--open round, and the two final signatures. DKD metadata contributes a fixed 81 bytes and the signatures contribute 128 bytes. Each sender contributes a 33-byte partial-DH term and a 98-byte DLEQ proof; the optional commit--open layer adds a 32-byte commitment and a 32-byte opening value per sender. Each receiver contributes a 33-byte DH term, a 98-byte proof, and a 33-byte public share. We exclude internal messages of the selected threshold-ECDSA protocol because their size is implementation-dependent.

Our results show totals of 1,334 bytes at $n{=}3$, 2,770 bytes at $n{=}7$, and 7,437 bytes at $n{=}20$. For the seven-member case, sender-side and receiver-side DSAG contribute 1,413 and 1,148 bytes, respectively, together accounting for more than 92\% of the total. The fixed DKD and signature fields remain small, so communication grows linearly with the number of participating members. A typical seven-member DAO therefore requires about 2.8\,KB of counted public data.

\subsection{Baseline Comparison (Fig.~\ref{fig:baseline-compare})}
\label{subsec:eval-baseline}

We compare our core with a standard single-user stealth-address payment and the centralized ECDSA check using the same secp256k1 implementation. The single-user SA path takes 3.88\,ms. The distributed \textsc{Dao$^2$} core takes 14.45, 27.91, and 54.42\,ms at $n=3,7,15$, while the ECDSA check stays near 2.5\,ms. The increase therefore comes from distributing DKD and DSAG across the participating members rather than from the final signature check.

Dividing the core cost by $n$ gives 4.82, 3.99, and 3.63\,ms at $n=3,7,15$, respectively. This pattern is consistent with a fixed setup component plus near-linear per-party work; it provides no indication of a superlinear coordination penalty over the measured range. Thus, the price of distributed control remains predictable for moderate DAO sizes. The comparison covers the implemented core only and excludes DLEQ and interactive threshold-ECDSA costs.

\subsection{Key-Derivation Depth Scalability (Fig.~\ref{fig:derivation-depth})}
\label{subsec:eval-depth}

We measure the next DKD step at chain depths 1, 10, 50, 100, 200, 500, and 1000, with $n{=}7$ and $t{=}2$. The cost remains between 2.318 and 2.547\,ms, with a mean of approximately 2.44\,ms. The 0.229\,ms peak-to-peak spread has no monotonic relation to depth, so the measurements show no cumulative slowdown as the derivation chain grows.

This result follows the protocol structure: a DKD step reads only the current parent state, obtains one HMAC-SHA512 output, and applies one additive update to each participating share. Its online cost is therefore constant in depth $k$ for fixed $n$, even though deriving all states from depth 1 through $k$ would still require $k$ sequential steps. Long-lived treasuries can consequently allocate later child keys without increasing the latency of each new derivation.


\begin{figure}[H]
\centering
\captionsetup{font=scriptsize,skip=2pt}

\begin{minipage}[t]{0.49\columnwidth}
\vspace{0pt}\centering
\begin{tikzpicture}
\pgfplotsset{width=0.61\linewidth,height=2.55cm,scale only axis}
\begin{axis}[
    ybar,
    bar width=3pt,
    xlabel={DAO size $n$},
    ylabel={Time (ms)},
    symbolic x coords={1,3,7,15},
    xtick=data,
    xticklabels={$n{=}1$,$n{=}3$,$n{=}7$,$n{=}15$},
    ymin=0, ymax=65,
    tick label style={font=\scriptsize},
    label style={font=\scriptsize},
    legend style={at={(0.02,0.98)},anchor=north west,
        font=\fontsize{5.0}{5.4}\selectfont,
        cells={anchor=west},row sep=-1pt,column sep=1pt,
        inner xsep=1pt,inner ysep=0.8pt},
    legend image post style={xscale=0.55},
    grid=major,
    grid style={dashed, gray!30},
    enlarge x limits=0.22,
]
\addplot[fill=cyan!60, draw=cyan!80] coordinates {
    (1,3.88) (3,0) (7,0) (15,0)
};
\addplot[fill=gray!50, draw=gray!70] coordinates {
    (1,0) (3,2.50) (7,2.48) (15,2.49)
};
\addplot[fill=red!55, draw=red!75] coordinates {
    (1,0) (3,14.45) (7,27.91) (15,54.42)
};
\legend{Standard SA,ECDSA check,\textsc{Dao$^2$} core}
\end{axis}
\end{tikzpicture}
\captionof{figure}{Core computation.}
\label{fig:baseline-compare}
\end{minipage}%
\hfill
\begin{minipage}[t]{0.49\columnwidth}
\vspace{0pt}\centering
\begin{tikzpicture}
\pgfplotsset{width=0.61\linewidth,height=2.55cm,scale only axis}
\begin{axis}[
    xlabel={Derivation depth $k$},
    ylabel={Time per derivation (ms)},
    xmode=log,
    log basis x=10,
    xtick={1,10,100,1000},
    xticklabels={1,10,100,1000},
    ymin=2.0, ymax=3.0,
    tick label style={font=\scriptsize},
    label style={font=\scriptsize},
    legend style={at={(0.98,0.98)},anchor=north east,
        font=\fontsize{5.0}{5.4}\selectfont,
        row sep=-1pt,column sep=1pt,
        inner xsep=1pt,inner ysep=0.8pt},
    legend image post style={xscale=0.55},
    grid=major,
    grid style={dashed, gray!30},
    mark size=1.8pt,
    every axis plot/.style={thick},
]
\addplot[color=blue, mark=*, mark options={fill=blue}] coordinates {
    (1,2.318) (10,2.427) (50,2.420) (100,2.488) (200,2.547) (500,2.453) (1000,2.430)
};
\addplot[color=red, dashed, domain=1:1000, samples=2] {2.44};
\legend{Measured,Mean ($\approx$2.44\,ms)}
\end{axis}
\end{tikzpicture}
\captionof{figure}{DKD cost versus depth~$k$.}
\label{fig:derivation-depth}
\end{minipage}%
\end{figure}

\section{Discussion}
\label{sec:discussion}

We discuss the main design choices, applications, limitations, and possible extensions of \textsc{Dao$^2$}.

\subsection{Design Rationale and Trade-offs}
\label{subsec:disc-design}

The three modules in \textsc{Dao$^2$} share one algebraic interface. DKD adds \(\omega^{(k)}\) to every receiver share, DSAG adds \(\rho^{(k)}\), and threshold ECDSA signs with the resulting states. Because Lagrange coefficients sum to one, both offsets pass through reconstruction and the public key changes by the matching group element.

Within DKD, we use the standard additive non-hardened BIP32 relation. The derivation offset is applied identically to every share, \(b_j^{(k)}=b_j^{(k-1)}+\omega^{(k)}\), so the Lagrange coefficients remain unchanged. The same rule then applies to the stealth offset \(\rho^{(k)}\). This common update rule simplifies the implementation and the correctness proof in Theorem~\ref{thm:correctness}.

We use a 2-out-of-$n$ threshold in our core evaluation. DKD and DSAG also support higher thresholds. To use them, the threshold-ECDSA scheme must support the additive-tweak interface in \S\ref{sec-2N}.

\subsection{Generality and Special-Case Reduction}
\label{subsec:disc-generality}

When \(n_1=n_2=t_1=t_2=1\), each DAO has a single member, so no threshold computation is needed. The receiver derives \(B^{(k)}\) locally. The sender computes \(\Omega^{(k)}=aB^{(k)}\), while the receiver computes the same shared secret as \(b^{(k)}A\). Both then derive the one-time address \(D^{(k)}=B^{(k)}+H(\Omega^{(k)}\parallel\xi^{(k)})G\). The protocol thus reduces to an ECDH-based payment with local BIP32 derivation. Threshold operations are needed only for distributed control.

\begin{revision}
The main practical setting is a 3--7 member DAO with threshold~2. The online protocol uses a constant number of broadcast rounds. Each partial DH term includes a non-interactive DLEQ proof. The optional sender commit-open layer adds one commitment round and one opening round.
\end{revision}

\subsection{Institutional Applications}
\label{subsec:disc-applications}

The downstream value of \textsc{Dao$^2$} is that payment privacy does not replace institutional control. Let the sender and receiver policies be defined by participant sets $[n_1]$ and $[n_2]$ with thresholds $t_1$ and $t_2$. A transaction selects working subsets $S_1$ and $S_2$ for stealth-address generation and detection, and signing subsets $T_1$ and $T_2$ satisfying $|T_1|\geq t_1$ and $|T_2|\geq t_2$. The receiver-facing output is
\[
D^{(k)}=B^{(k)}+H(\Omega^{(k)}\parallel\xi^{(k)})G,
\]
yet redemption remains threshold-controlled because the parties hold shares of the corresponding key rather than reconstructing it. The blockchain verifies ordinary ECDSA signatures under $A$ for payment and $D^{(k)}$ for redemption; the organizational policy remains off chain.

For DAO treasuries, this interface supports treasury-to-treasury grants, investments, or vendor payments in which both sides retain their existing approval rules while the receiver obtains a fresh unlinkable destination. For consortium-controlled funds, the participants may represent different member organizations rather than individual wallet users. The long-term key $A$ fixes the sender policy, while $B^{(k)}$ evolves deterministically for receiver continuity; qualified subsets can execute individual transfers, and DLEQ verification prevents a participating member from contributing a Diffie--Hellman term inconsistent with its registered public share.

For custodial infrastructures, DKD maps naturally to per-account or per-deposit child states, while DSAG hides the public link between an institutional root wallet and a particular incoming payment. The same construction can support MPC custodians, shared escrow services, and institutional settlement wallets because no operator needs the master secret or the one-time spending key in the clear. The session tag $\mathsf{id}^{(k)}$ and evolving chaincode $cc^{(k)}$ provide deterministic internal indexing, whereas $\xi^{(k)}$ and $D^{(k)}$ provide transaction-level unlinkability. Operational policy, audit logging, and regulatory identity controls remain deployment-layer responsibilities rather than cryptographic guarantees of the protocol.

\subsection{Limitations}
\label{subsec:disc-limitations}

Our prototype measures the secp256k1 core and counted public data, but not a complete threshold-ECDSA implementation, on-chain fees, or contract execution. UTXO systems map naturally to one-time outputs, whereas account-based systems require a contract representation of $D^{(k)}$. The protocol uses $O(1)$ broadcast rounds, but still requires a threshold number of online parties on each side.

Theorem~\ref{thm:privacy} covers the public blockchain transcript, not timing, IP-address, or participation-pattern leakage. Network anonymity is therefore a deployment requirement.

\subsection{Extensions and Open Problems}
\label{subsec:disc-future}

The current design assumes fixed membership and threshold within each epoch. Proactive resharing could update both while preserving wallet continuity, and batched or pre-computed derivations could reduce repeated work without introducing cross-transaction links.

Cross-chain deployment requires both chains to use consistent derivation states and a bridge that keeps the destination funds under threshold control. Confidential transactions or decoy mechanisms could provide additional privacy, but combining them with DSAG requires further security analysis. Before deployment, the implementation also needs formal verification and constant-time code.

\section{Related Work}
\label{sec:rw}

\smallskip
\noindent\textbf{Blockchain DAOs.}
DAOs have been studied from conceptual, system, and empirical perspectives. Wang et al.~\cite{wang2019dao} clarified the basic concept and application model of DAOs, while Liu et al.~\cite{liu2021dao} summarized the broader technical and social landscape of blockchain-based DAO systems. More recent work examined DAO architectures in Web3 applications \cite{yu2023leveraging}, surveyed DAO governance and tooling \cite{tang2025decentralized}, and analyzed governance behavior in deployed DAO ecosystems \cite{bellavitis2023dao,wang2025understanding}. This literature explains why digital assets are increasingly controlled by organizations rather than single users. However, its emphasis is on governance and system design, not on cryptographic transaction mechanisms between threshold-controlled organizations \cite{fabrega2025voting,sharma2024unpacking,feichtinger2023hidden,feichtinger2024sok}.

\smallskip
\noindent\textbf{Blockchain threshold signatures.}
Threshold signatures are the main cryptographic tool for shared control of blockchain assets. Foundational work on distributed key generation was developed by Gennaro et al.~\cite{gennaro2007dkg}, and practical threshold ECDSA was substantially advanced by Lindell and Nof~\cite{lindell2018fast}. Then, the literature has emphasized deployability. Canetti et al.~\cite{canetti2020uc} introduced proactive and non-interactive threshold ECDSA with identifiable aborts, and Doerner et al.~\cite{doerner2024threshold} reduced threshold ECDSA to three rounds. Recent work has also strengthened the setup layer and blockchain-facing service model: Kate et al.~\cite{kate2024nivss} studied non-interactive verifiable secret sharing and its application to distributed key generation, while Thyagarajan et al.~\cite{thyagarajan2025vitarit} considered how threshold services can be paid on Bitcoin-like systems. These advances are directly relevant to our authorization backbone, but they still focus on signing, setup, or service payment rather than anonymous receipt by a threshold-controlled recipient.

\smallskip
\noindent\textbf{Stealth-address privacy.}
Receiver privacy in cryptocurrencies originates from one-time-address and ring-based designs~\cite{liu2005linkable}. CryptoNote introduced the scan/spend-key stealth-address paradigm for unlinkable receiving \cite{vanSaberhagen2013cryptonote,sun2017ringct,yuen2020ringct,duan2024concise}, and RingCT later strengthened privacy by hiding transferred amounts in Monero-style systems \cite{noether2016ringct}. Recent work has broadened this direction. Glaeser et al.~\cite{glaeser2022foundations} formalized the security foundations of coin-mixing services, and Pu et al.~\cite{pu2023stealth} gave a more explicit cryptographic treatment of stealth-signature-style private payments. Wang et al.~\cite{wang2024dsag} moved closer to our setting by considering distributed stealth-address generation for threshold signatures. Nevertheless, these works study privacy primitives or closely related subprotocols in isolation. They do not provide an end-to-end framework in which anonymous receipt, threshold redemption, and organizational wallet continuity are handled jointly.

\smallskip
\noindent\textbf{Distributed key derivation.}
BIP32 remains the standard basis for hierarchical deterministic wallet derivation \cite{bip32}. It provides an effective way to derive fresh addresses and child keys under a single root secret, but the original model assumes local control of that secret and therefore does not directly transfer to threshold-managed organizational wallets. Zhong et al.~\cite{zhong2023dkd} are directly relevant because they adapt key derivation to multi-party blockchain asset management. Even so, distributed derivation remains less developed than threshold signing or payment privacy. Existing work explains how child keys can be derived without reconstructing a master secret, but it does not fully address how anonymously received outputs should be bound to derivation state and later redeemed under threshold authorization inside a DAO treasury.

\section{Conclusion}
\label{sec-conclusion}

We presented \textsc{Dao$^2$}, a privacy-adaptable protocol that enables direct asset transfers between threshold-controlled DAOs. \textsc{Dao$^2$} supports both non-anonymous transfers, where the receiver derives a standard child address via distributed key derivation (DKD), and anonymous transfers, where an unlinkable one-time stealth destination is produced via distributed stealth-address generation (DSAG), all without reconstructing any master secret. At its core, \textsc{Dao$^2$} couples three cryptographic primitives: DKD for hierarchical wallet continuity under threshold control, DSAG for recipient privacy and unlinkability, and threshold signatures as the unified authorization backbone for both payment and redemption.

We formally prove security under a sub-threshold adaptive adversary and give a forward-secrecy argument under one-time-share erasure. Our prototype confirms practicality.




\appendices
\section{Notation}
\label{app:notation}

\begin{table}[H]
\centering
\caption{Summary of notation.}\label{tab:notation}
\footnotesize
\setlength{\tabcolsep}{4pt}
\renewcommand{\arraystretch}{1.18}
\resizebox{\linewidth}{!}{
\begin{tabular}{@{} r|l @{}}
\toprule
\multicolumn{1}{c|}{\textbf{Symbol}} &
\multicolumn{1}{c}{\textbf{Description}} \\
\midrule
$\mathbb{G},\; q,\; G$ & Prime-order elliptic-curve group, its order, and generator (secp256k1) \\
$H(\cdot)$;\; $\text{HMAC-SHA512}$ & Hash to $\mathbb{Z}_q$ (SHA-256 based);\; keyed hash for BIP32-style derivation \\
$\lambda_{i,S}$ & Lagrange coefficient for index $i$ w.r.t.\ subset $S$ \\
\midrule
$\mathsf{DAO}_1 / \mathsf{DAO}_2$ & Sender / receiver organization with $n_1$/$n_2$ members and threshold $t_1$/$t_2$ \\
$S_1,\, S_2$ & Qualified subsets for DSAG generation and recovery \\
$T_1,\, T_2$ & Signing subsets: $T_1\!\subseteq\![n_1]$, $T_2\!\subseteq\!S_2$, and $|T_\ell|\!\ge\!t_\ell$ \\
$a,\; a_i,\; A{=}aG$ & Sender aggregate secret, $i$-th share, and long-term public key \\
$b^{(k)},\; b_j^{(k)},\; B^{(k)}{=}b^{(k)}G$ & Receiver child secret at epoch $k$, $j$-th share, and child public key \\
$cc^{(k)},\; \mathsf{id}^{(k)}$ & Chain code and public derivation tag at epoch $k$ \\
\midrule
$\omega^{(k)}$ & DKD offset: first 256 bits of $\text{HMAC-SHA512}(cc^{(k-1)},\, B^{(k-1)} \!\parallel\! \mathsf{id}^{(k)})$ \\
\midrule
$\Omega_i^{(k)} {=} a_i B^{(k)}$ & Sender $i$'s partial DH term \\
$\Omega^{(k)} {=} aB^{(k)}$ & Distributed shared secret (sender); $\Omega'^{(k)}{=}b^{(k)}A$ (receiver; $=\Omega^{(k)}$) \\
$\pi_i^S,\;\pi_j^R$ & DLEQ proofs for sender and receiver partial DH terms \\
$\xi^{(k)},\; \rho^{(k)} {=} H(\Omega^{(k)} \!\parallel\! \xi^{(k)})$ & Random public label;\; stealth offset scalar \\
$D^{(k)} {=} B^{(k)} {+} \rho^{(k)}G$ & One-time stealth destination \\
$d_j^{(k)} {=} b_j^{(k)} {+} \rho^{(k)}$ & One-time secret share;\; $D_j^{(k)}{=}d_j^{(k)}G$: public share \\
\midrule
$\mathsf{TS.Tweak},\;\mathsf{TS.Sign}_S,\; \mathsf{TS.Verify}$ & Additive key update and threshold-ECDSA interfaces \\
$\sigma_{\mathrm{pay}}^{(k)},\; \sigma_{\mathrm{spend}}^{(k)}$ & Signatures on payment and redemption transactions \\
$\delta^{(k)}$ & Session descriptor $(B^{(k)}, cc^{(k)}, \mathsf{id}^{(k)})$ \\
$\tau_{\mathrm{chain}}^{(k)}$ & Public chain transcript $(m_{\mathrm{pay}}^{(k)}, \sigma_{\mathrm{pay}}^{(k)}, D^{(k)}, \mathsf{id}^{(k)}, \xi^{(k)})$ \\
\bottomrule
\end{tabular}
}
\end{table}

\section{Security Proofs}
\label{sec:security-proofs}

\begin{revision}
We use the notation of \S\ref{sec-preli}--\S\ref{sec:e2e-protocol} and the assumptions stated in each theorem.
\end{revision}

\subsection{Proof of Theorem~\ref{thm:correctness} (Transaction Correctness)}
\label{prf:correctness}

Correctness requires that the stealth output $D^{(k)}$ published on chain is always redeemable by the receiver DAO, and that all honest parties converge to a consistent post-transaction wallet state. Because the \textsc{Dao$^2$} protocol chains three algebraic operations (distributed key derivation, distributed shared-secret computation, and one-time key construction), the proof must show that each operation preserves the Lagrange reconstruction invariant. We establish this through three lemmas, one per operation, and then combine them.

\begin{lemma}[Distributed key derivation consistency]
\label{lem:dkd-consistency}
For any derivation step $k$, if all honest parties in $S_2$ hold consistent parent shares $\{b_j^{(k-1)}\}_{j \in S_2}$ with $\sum_{j \in S_2} \lambda_{j,S_2} b_j^{(k-1)} = b^{(k-1)}$, then after deriving the child state:
\[
b_j^{(k)} = b_j^{(k-1)} + \omega^{(k)}, \quad B^{(k)} = B^{(k-1)} + \omega^{(k)}G,
\]
it holds that $\sum_{j \in S_2} \lambda_{j,S_2} b_j^{(k)} = b^{(k-1)} + \omega^{(k)} = b^{(k)}$ and $B^{(k)} = b^{(k)} G$.
\end{lemma}

\begin{prf}[Lemma~\ref{lem:dkd-consistency}]
The offset $\omega^{(k)}$ is the first 256 bits of
$\text{HMAC-SHA512}(cc^{(k-1)},\, B^{(k-1)} \!\parallel\! \mathsf{id}^{(k)})$.
Since $B^{(k-1)}$, $cc^{(k-1)}$, and $\mathsf{id}^{(k)}$ are public, all honest parties compute the same $\omega^{(k)}$. Adding a uniform scalar to every share preserves the Lagrange reconstruction relation:
\begin{align}
\sum_{j \in S_2} \lambda_{j,S_2} b_j^{(k)}
&= \sum_{j \in S_2} \lambda_{j,S_2} \bigl(b_j^{(k-1)} + \omega^{(k)}\bigr) \nonumber\\
&= b^{(k-1)} + \omega^{(k)} \!\underbrace{\sum_{j \in S_2} \lambda_{j,S_2}}_{=\,1}
= b^{(k)}.  \label{eq:dkd-lagrange}
\end{align}
Multiplying both sides by $G$ yields
$B^{(k)} = b^{(k)}G$.
\qed
\end{prf}

\begin{lemma}[Shared-secret consistency]
\label{lem:shared-secret-consistency}
The sender-side and receiver-side computations yield the same shared secret:
$\Omega^{(k)} = aB^{(k)} = b^{(k)}A = \Omega'^{(k)}$.
\end{lemma}

\begin{prf}[Lemma~\ref{lem:shared-secret-consistency}]
On the sender side, each party $P_i$ computes $a_i B^{(k)}$ locally; the results are then combined via Lagrange coefficients:
\begin{align}
\Omega^{(k)}
&= \sum_{i \in S_1} \lambda_{i,S_1}\, (a_i B^{(k)})
= \Bigl(\sum_{i \in S_1} \lambda_{i,S_1}\, a_i\Bigr) B^{(k)} \nonumber\\
&= a \cdot B^{(k)} = a\, b^{(k)} G. \label{eq:omega-sender}
\end{align}
On the receiver side:
\begin{align}
\Omega'^{(k)}
&= \sum_{j \in S_2} \lambda_{j,S_2}\,(b_j^{(k)} A)
= \Bigl(\sum_{j \in S_2} \lambda_{j,S_2}\, b_j^{(k)}\Bigr) A \nonumber\\
&= b^{(k)} \cdot A = b^{(k)}\, a\, G. \label{eq:omega-receiver}
\end{align}
Since $a\,b^{(k)} G = b^{(k)}\,a\,G$ by the commutativity of scalar multiplication in $\mathbb{G}$, we conclude $\Omega^{(k)} = \Omega'^{(k)}$.
\qed
\end{prf}

\begin{lemma}[One-time key reconstruction]
\label{lem:otk-reconstruction}
Let $\rho^{(k)} = H(\Omega^{(k)} \parallel \xi^{(k)})$ and $d_j^{(k)} = b_j^{(k)} + \rho^{(k)}$. Then:
\[
D^{(k)} = \Big(\sum_{j \in S_2} \lambda_{j,S_2} d_j^{(k)}\Big) G.
\]
\end{lemma}

\begin{prf}[Lemma~\ref{lem:otk-reconstruction}]
Expanding the Lagrange sum:
\[
\sum_{j \in S_2} \lambda_{j,S_2} d_j^{(k)}
= \sum_{j \in S_2} \lambda_{j,S_2} (b_j^{(k)} + \rho^{(k)})
= b^{(k)} + \rho^{(k)}.
\]
From the sender-side definition:
$D^{(k)} = B^{(k)} + \rho^{(k)} G = b^{(k)} G + \rho^{(k)} G = (b^{(k)} + \rho^{(k)}) G$.
Hence $D^{(k)} = (\sum_j \lambda_{j,S_2} d_j^{(k)}) G$.
\qed
\end{prf}

\noindent\textit{Proof of Theorem~\ref{thm:correctness}.}
Combining Lemmas~\ref{lem:dkd-consistency}--\ref{lem:otk-reconstruction}: The derivation preserves the sharing structure (Lemma~\ref{lem:dkd-consistency}), both sides compute the same shared secret (Lemma~\ref{lem:shared-secret-consistency}), and the resulting one-time shares reconstruct the correct one-time public key (Lemma~\ref{lem:otk-reconstruction}). Therefore the receiver can invoke $\mathsf{TS.Sign}_{T_2}$ under public key $D^{(k)}$ using shares $\{d_j^{(k)}\}_{j \in T_2}$, producing a valid signature by the correctness of the threshold-signature instantiation (\S\ref{sec-2N}).

\smallskip
\noindent\textit{Remark (threshold-signing compatibility).}
\begin{revision}
The shares $\{d_j^{(k)}\}$ form a valid Shamir sharing of $d^{(k)}$, with $D_j^{(k)}=B_j^{(k)}+\rho^{(k)}G$. Before redemption, each party applies the additive-tweak adapter in \S\ref{sec-2N}, updates its public-key-dependent state, and creates fresh preprocessing for $D^{(k)}$. Key-independent setup may be reused, but nonces and presignatures are never reused across keys. The result is a standard ECDSA signature under $D^{(k)}$; the DKD and DSAG equations do not change.
\end{revision}

\begin{revision}
The agreed inputs determine the same public state \((B^{(k)},cc^{(k)})\) and tweak. Honest parties obtain consistent, party-specific child shares and compatible ECDSA states. They commit this successor after the same redemption transaction is confirmed, making the next derivation well-defined.
\end{revision}
\qed

\subsection{Proof of Theorem~\ref{thm:threshold-spending} (Spending Unforgeability)}
\label{prf:threshold-spending}

The protected keys are the sender key \(a\), receiver parent keys, and their derived spending keys. We first recall the secrecy of Shamir setup shares, then reduce unauthorized spending to the full signing-interface requirement. The latter reduction allows the adversary to know public tweaks and DH values disclosed in sessions it participates in.

\begin{lemma}[Sender-side key secrecy]
\label{lem:sender-secrecy}
Under DL, the setup view consisting of \(A\), public Shamir shares, and fewer than \(t_1\) scalar shares at fixed corrupted indices does not permit recovery of \(a\).
\end{lemma}

\begin{prf}[Lemma~\ref{lem:sender-secrecy}]
Given a DL challenge \(A=aG\), choose uniform scalar values at the corrupted indices and enough additional indices to obtain \(t_1-1\) points. Interpolation with the unknown constant \(a\) defines a uniformly shared key. Every share has the form \(\alpha_i a+\beta_i\), with known coefficients, so its public point is computable as \(\alpha_iA+\beta_iG\); corrupted scalar shares are known by construction. This simulates the setup view without knowing \(a\). Recovering \(a\) from that view solves the challenge. This lemma concerns the setup view; the full protocol view is covered by the signing-interface requirement.
\qed
\end{prf}

\begin{lemma}[One-time key secrecy]
\label{lem:otk-secrecy}
\begin{revision}
Under the full interface requirement in \S\ref{sec-2N}, fixed sub-threshold corrupted sets cannot recover an honestly protected base or one-time key from the permitted joint protocol view.
\end{revision}
\end{lemma}

\begin{prf}[Lemma~\ref{lem:otk-secrecy}]
\begin{revision}
If an adversary recovers a scalar \(x\) for a protected key \(X=xG\), it can run ordinary ECDSA to sign a fresh message that was never authorized under \(X\). This breaks the interface's multi-key unforgeability requirement. The reduction also applies to \(d^{(k)}\), regardless of whether \(\rho^{(k)}\) is known.
\end{revision}
\qed
\end{prf}

\noindent\textit{Proof of Theorem~\ref{thm:threshold-spending}.}
Condition on consistent setup and successful verification of only valid DH contributions. The outer protocol performs public deterministic derivations and invokes the signing interface with the resulting tweaks. Its corrupted-party view, including accepted DH terms and proofs, is among the joint views covered by that interface requirement. Forward all authorized signing requests with their session and key identifiers. An unauthorized output signature wins the interface forgery experiment directly; recovery of a scalar key yields a fresh forgery by Lemma~\ref{lem:otk-secrecy}.

Let \(\epsilon_{\mathrm{TS}}\) bound interface forgery over all polynomially many sessions and related keys, and let \(\epsilon_{\mathrm{chk}}\) bound invalid setup or proof acceptance. Then
\[
\Pr[\mathsf{Exp}^{\mathrm{spend\text{-}uf}}_{\mathcal A}(1^\kappa)=1]
\leq \epsilon_{\mathrm{TS}}+\epsilon_{\mathrm{chk}}
=\mathsf{negl}(\kappa).
\]
\qed

\subsection{Proof of Theorem~\ref{thm:privacy} (Recipient Indistinguishability)}
\label{prf:privacy}

We prove privacy for the external observer of Definition~3. It sees the two public root states and the signed payment transcripts, but no private shares or off-chain session records. The payment fields other than the destination are the same in both challenge cases. We use independent random oracles \(H\) for stealth offsets and \(H_{\mathrm{sig}}\) for ECDSA message digests.

\begin{lemma}[Public-destination indistinguishability]
\label{lem:stealth-prf}
Under DDH and the random-oracle model for \(H\), the public destination in Definition~3 hides the chosen receiver, even given both candidate public root states and \(A\).
\end{lemma}

\begin{prf}[Lemma~\ref{lem:stealth-prf}]
For each candidate lineage, write its advertised parent as \(B_c^{(0)}\). All subsequent offsets are computable from this public state and the tags, so
\[
B_c^{(k)}=B_c^{(0)}+\Delta_{c,k}G,\qquad
aB_c^{(k)}=aB_c^{(0)}+\Delta_{c,k}A.
\]
Here \(\Delta_{c,k}\) is the sum of the intervening DKD offsets. Embed a DDH tuple \((G,A,B_c^{(0)},T_c)\) at one root. If \(T_c=aB_c^{(0)}\), using \(T_c+\Delta_{c,k}A\) gives the real DH values. If \(T_c\) is uniform, it gives the root-replacement hybrid. The other root can be generated with a known scalar, allowing its DH values to be computed from \(A\). Thus this transition is justified by DDH even though descendant keys are related.

In the random-root hybrid, a query at the hidden input
\((T_c+\Delta_{c,k}A)\parallel\xi^{(k)}\) is a guess of \(T_c\). Before any such hit, its hash value is fresh and uniform. Consequently \(D_c^{(k)}=B_c^{(k)}+H((T_c+\Delta_{c,k}A)\parallel\xi^{(k)})G\) is uniform and independent of the receiver bit. For one session and \(Q_H\) queries, a hit has probability at most \(Q_H/q\). The uniform-destination transition is therefore statistical up to this bad event, rather than perfect. Applying the root replacement to either candidate proves the lemma.
\qed
\end{prf}

\begin{lemma}[Signed-transcript indistinguishability]
\label{lem:public-transcript-reduction}
Under DDH, independent random oracles \(H,H_{\mathrm{sig}}\), and honest randomized-ECDSA outputs, the complete signed payment transcript in Definition~3 hides the receiver.
\end{lemma}

\begin{prf}[Lemma~\ref{lem:public-transcript-reduction}]
The sender key is also used in DH, so its signature cannot be treated as public post-processing of \(D\). We instead simulate its output using only \(A\). For a new payment message \(m\), choose \(u\in_R\mathbb Z_q\) and \(v\in_R\mathbb Z_q^*\), and set
\[
R=uG+vA,\quad r=x(R)\bmod q,\quad
s=rv^{-1},\quad h=us.
\]
Retry if \(R\) is the identity or \(r=0\). If \(H_{\mathrm{sig}}(m)\) is unset, program it to \(h\) and return \((r,s)\). Verification holds because
\[
h s^{-1}G+r s^{-1}A=uG+vA=R.
\]
For the unknown scalar \(a\), put \(\nu=u+va\). Then \(R=\nu G\) and \(s=\nu^{-1}(h+ra)\), the ordinary ECDSA signing equation. The change of variables between \((u,v)\) and \((\nu,h)\) gives the honest random-nonce, random-digest distribution up to negligible differences from zero-value retries.

Abort the simulation if \(H_{\mathrm{sig}}(m)\) was already queried. In the random-root hybrid, each new message contains a fresh uniform destination until a hidden-DH query occurs. Among \(Q_s\) payments and \(Q_{\mathrm{sig}}\) digest queries, a prior-message hit or message repetition has probability
\(O((Q_sQ_{\mathrm{sig}}+Q_s^2)/q)\).
The DDH reduction runs this same simulator in both tuple cases, including its abort rule. Its outputs and bad-event flag are efficient functions of the supplied tuple. Thus the abort probability in the real-root case differs from that in the random-root case by at most the DDH advantage. Before the first abort, the simulated and honest signature views have the same distribution up to the retry term.

All other payment fields are fixed across the receiver choices. Combining the DDH replacement, signature simulation, and the negligible bad events proves signed-transcript indistinguishability without knowing \(a\).
\qed
\end{prf}

\begin{lemma}[Multi-transaction unlinkability]
\label{lem:multi-tx-unlink}
The external observer has negligible advantage in the multi-payment experiment of Definition~3.
\end{lemma}

\begin{prf}[Lemma~\ref{lem:multi-tx-unlink}]
Replace each challenge lineage's root DH value once. For every descendant on that lineage, use the same hidden root \(T_c\) and the value \(T_c+\Delta_{c,k}A\). This preserves the algebraic relations between sessions; we never replace related DH values independently. Each transition is a DDH reduction with the signature simulator of Lemma~\ref{lem:public-transcript-reduction}.

Fresh labels give distinct stealth-hash inputs except with birthday probability \(O(Q_s^2/2^{256})\). Before a hidden-input hit, their random-oracle outputs are independent, even though the underlying DH points are related. Hence all destinations are independently uniform in the final hybrid, whether the payments target one lineage or two. The simulated signatures and receiver-independent message fields preserve this equality of distributions. For \(L\) replaced roots, the total distinguishing advantage is bounded by
\[
O\!\left(L\epsilon_{\mathrm{DDH}}
+\frac{Q_s(Q_H+Q_{\mathrm{sig}}+Q_s)}{q}
+\frac{Q_s^2}{2^{256}}\right),
\]
including the negligible ECDSA retry terms. All query and session counts are polynomial, so this bound is negligible.
\qed
\end{prf}

\noindent\textit{Proof of Theorem~\ref{thm:privacy}.}
Lemma~\ref{lem:public-transcript-reduction} establishes the single-payment claim with the actual signed chain view. Lemma~\ref{lem:multi-tx-unlink} establishes unlinkability for the same observer over multiple payments. Neither claim includes an observer that obtained session DH values or receiver-identifying off-chain records.
\qed

\subsection{Proof of Theorem~\ref{thm:robustness} (Byzantine Robustness)}
\label{prf:robustness}

Robustness concerns acceptance safety under the agreed-session model. We show that invalid contributions are rejected and that honest parties accepting the same confirmed session commit the same successor state. A failed check causes a visible abort; message delays or insufficient participation may prevent completion.

\begin{lemma}[Setup robustness]
\label{lem:dkg-robust}
If the initial threshold keys of $\mathsf{DAO}_1$ and $\mathsf{DAO}_2$ are generated by a robust verifiable DKG, then fewer than the corresponding threshold of malicious participants cannot cause honest parties to output inconsistent shares or an invalid aggregate public key.
\end{lemma}

\begin{prf}[Lemma~\ref{lem:dkg-robust}]
This lemma is inherited from the setup procedure. In a Feldman-style DKG, dealer \(P_i\) samples \(f_i(z)=\sum_{\ell=0}^{t-1}c_{i,\ell}z^\ell\), publishes commitments \(C_{i,\ell}=c_{i,\ell}G\), and sends \(s_{ij}=f_i(j)\) to \(P_j\). Every honest receiver verifies
\[
s_{ij}G\stackrel{?}{=}\sum_{\ell=0}^{t-1}j^\ell C_{i,\ell}.
\]
For \(t=2\), this becomes \(s_{ij}G=C_{i,0}+jC_{i,1}\). A failed check triggers a complaint and excludes the faulty dealer. Any robust verifiable DKG with these properties suffices; the online protocol uses only its consistent shares and aggregate public key.
\qed
\end{prf}

\begin{lemma}[Derivation state consistency]
\label{lem:derivation-consistency}
For the same agreed parent state and derivation tag, all honest parties compute the same child public key \(B^{(k)}\) and chaincode \(cc^{(k)}\).
\end{lemma}

\begin{prf}[Lemma~\ref{lem:derivation-consistency}]
Input agreement is a premise of \S\ref{sec-threat-model}, rather than a consequence of deterministic hashing. Given those common inputs, HMAC-SHA512 and the additive update yield identical public child states. Parties commit only after the same agreed transaction is confirmed. Induction therefore preserves state agreement across committed sessions.
\qed
\end{prf}

\begin{lemma}[Partial-DH consistency]
\label{lem:partial-dh-consistency}
\begin{revision}
Except with negligible probability, every accepted sender term equals \(a_iB^{(k)}\), and every accepted receiver term equals \(b_j^{(k)}A\).
\end{revision}
\end{lemma}

\begin{prf}[Lemma~\ref{lem:partial-dh-consistency}]
\begin{revision}
DLEQ verification links \((G,A_i)\) to \((B^{(k)},\Omega_i^{(k)})\). Since \(A_i=a_iG\), knowledge soundness gives \(\Omega_i^{(k)}=a_iB^{(k)}\). The same argument gives \(\Omega_j'^{(k)}=b_j^{(k)}A\) on the receiver side. Invalid terms are rejected before aggregation.
\end{revision}
\qed
\end{prf}

\begin{lemma}[One-time share verification]
\label{lem:share-verification}
\begin{revision}
Checking each $D_j^{(k)}$ before the aggregate equation identifies any party that publishes an inconsistent share.
\end{revision}
\end{lemma}

\begin{prf}[Lemma~\ref{lem:share-verification}]
\begin{revision}
The expected point $B_j^{(k)}+\rho^{(k)}G$ is public after detection. Any different $\widetilde D_j^{(k)}$ fails the per-share check and identifies $j$. If all checks pass, linearity and $\sum_j\lambda_{j,S_2}=1$ give
\[
\sum_{j\in S_2}\lambda_{j,S_2}D_j^{(k)}
=B^{(k)}+\rho^{(k)}G=D^{(k)}.
\]
The threshold-ECDSA layer detects or identifies any later scalar contribution that does not match the accepted public share.
\end{revision}
\qed
\end{prf}

\begin{lemma}[Threshold-signing robustness]
\label{lem:signing-robust}
Under the signing-interface requirement, invalid signer contributions cannot produce an invalid accepted signature; the interface may instead abort visibly.
\end{lemma}

\begin{prf}[Lemma~\ref{lem:signing-robust}]
\begin{revision}
The interface requirement supplies this property. The outer protocol checks the resulting signature and commits no successor state on failure.
\end{revision}
\qed
\end{prf}

\noindent\textit{Proof of Theorem~\ref{thm:robustness}.}
We verify each condition of Definition~4 (\S\ref{sec:security-goals}):
\begin{itemize}
\item \qw{\textit{Invalid contribution acceptance}: Robust setup, deterministic derivation, DLEQ proofs, and per-share checks prevent acceptance (Lemmas~\ref{lem:dkg-robust}--\ref{lem:share-verification}).}
\item \textit{Divergent accepted state}: The agreed record and common confirmation rule give honest acceptances the same deterministic successor (Lemma~\ref{lem:derivation-consistency}). Failed checks or missing responses cause no commit.
\end{itemize}
Hence $\Pr[\mathsf{Exp}^{\mathrm{rob}}_{\mathcal{A}}(1^\kappa) = 1] \leq \mathsf{negl}(\kappa)$.
\qed


\begin{thebibliography}{10}

\bibitem{wang2025understanding}
Qin Wang, Guangsheng Yu, Yilin Sai, Caijun Sun, Lam~Duc Nguyen, and Shiping Chen.
\newblock Understanding {DAO}s: An empirical study on governance dynamics.
\newblock {\em IEEE TCSS}, 12(5):2814--2832, 2025.

\bibitem{tang2025decentralized}
Caiyan Tang, Qi~Cai, Chengzu Dong, et~al.
\newblock Decentralised autonomous organizations ({DAO}s): An exploratory survey.
\newblock {\em ACM DLT)}, pages 1--25, 2025.

\bibitem{lindell2018fast}
Yehuda Lindell and Ariel Nof.
\newblock Fast secure multiparty ecdsa with practical distributed key generation and applications to cryptocurrency custody.
\newblock In {\em ACM CCS}, pages 1837--1854, 2018.

\bibitem{canetti2020uc}
Ran Canetti, Rosario Gennaro, Steven Goldfeder, Nikolaos Makriyannis, and Udi Peled.
\newblock Uc non-interactive, proactive, threshold ecdsa with identifiable aborts.
\newblock In {\em ACM CCS}, pages 1769--1787, 2020.

\bibitem{doerner2024threshold}
Jack Doerner, Yashvanth Kondi, Eysa Lee, and abhi Shelat.
\newblock Threshold {ECDSA} in three rounds.
\newblock In {\em IEEE SP}, pages 3053--3071, 2024.

\bibitem{vanSaberhagen2013cryptonote}
Nicolas V.~Saberhagen.
\newblock {CryptoNote v2.0}.
\newblock \url{https://www.getmonero.org/resources/research-lab/pubs/cryptonote-whitepaper.pdf}, 2013.
\newblock Whitepaper.

\bibitem{wang2024dsag}
Yujue Wang, Lin Zhong, Jun Du, Yudi Zou, Kevin He, and Andrew Zhang.
\newblock A distributed stealth address generation protocol for threshold signatures.
\newblock In {\em IEEE ISPA}, pages 2014--2021, 2024.

\bibitem{pu2023stealth}
Sihang Pu, Sri Aravinda~Krishnan Thyagarajan, Nico Doettling, and Lucjan Hanzlik.
\newblock Post quantum fuzzy stealth signatures and applications.
\newblock In {\em ACM CCS}, pages 371--385, 2023.

\bibitem{bip32}
Pieter Wuille.
\newblock {BIP-0032: Hierarchical Deterministic Wallets}.
\newblock \url{https://github.com/bitcoin/bips/blob/master/bip-0032.mediawiki}, 2012.
\newblock Bitcoin Improvement Proposal 32, assigned 2012-02-11.

\bibitem{zhong2023dkd}
Lin Zhong, Yujue Wang, Yong Ding, Jun Du, Kevin He, and Andrew Zhang.
\newblock Distributed key derivation for multi-party management of blockchain digital assets.
\newblock In {\em IEEE ICPADS}, pages 715--720, 2023.

\bibitem{zhong2023fast2n}
Lin Zhong, Yujue Wang, Jun Du, Daji Liang, Ziyuan Zhong, Kevin He, and Andrew Zhang.
\newblock Fast 2-out-of-n ecdsa threshold signature.
\newblock In {\em IEEE ISPA}, pages 456--465, 2023.

\bibitem{rivest2001ring}
Ronald~L. Rivest, Adi Shamir, and Yael Tauman.
\newblock How to leak a secret.
\newblock In {\em ASIACRYPT}, volume 2248, pages 552--565. Springer, 2001.

\bibitem{noether2016ringct}
Shen Noether, Adam Mackenzie, and {The Monero Research Lab}.
\newblock Ring confidential transactions.
\newblock {\em Ledger}, 1:1--18, 2016.

\bibitem{gennaro2007dkg}
Rosario Gennaro, Stanislaw Jarecki, Hugo Krawczyk, and Tal Rabin.
\newblock Secure distributed key generation for discrete-log based cryptosystems.
\newblock {\em Journal of Cryptology}, 20(1):51--83, 2007.

\bibitem{kate2024nivss}
Aniket Kate, Easwar~Vivek Mangipudi, Pratyay Mukherjee, Hamza Saleem, and Sri Aravinda~Krishnan Thyagarajan.
\newblock Non-interactive {VSS} using class groups and application to {DKG}.
\newblock In {\em ACM CCS}, 2024.

\bibitem{wang2019dao}
Shuai Wang, Wenwen Ding, Juanjuan Li, Yong Yuan, Liwei Ouyang, and Fei-Yue Wang.
\newblock Decentralized autonomous organizations: Concept, model, and applications.
\newblock {\em IEEE TCSS}, 6(5):870--878, 2019.

\bibitem{liu2021dao}
Lu~Liu, Sicong Zhou, Huawei Huang, and Zibin Zheng.
\newblock From technology to society: An overview of blockchain-based {DAO}.
\newblock {\em IEEE Open Journal of the Computer Society}, 2:204--215, 2021.

\bibitem{yu2023leveraging}
Guangsheng Yu, Qin Wang, Tingting Bi, Shiping Chen, and Xiwei Xu.
\newblock Leveraging architectural approaches in web3 applications-a dao perspective focused.
\newblock In {\em ICBC}, 2023.

\bibitem{bellavitis2023dao}
Cristiano Bellavitis, Christian Fisch, and Paul~P. Momtaz.
\newblock The rise of decentralized autonomous organizations ({DAO}s): a first empirical glimpse.
\newblock {\em Venture Capital}, 25(2):187--203, 2023.

\bibitem{fabrega2025voting}
Andres Fabrega, Amy Zhao, Jay Yu, James Austgen, Sarah Allen, Kushal Babel, Mahimna Kelkar, and Ari Juels.
\newblock $\{$Voting-Bloc$\}$ entropy: A new metric for $\{$DAO$\}$ decentralization.
\newblock In {\em USENIX Security}, pages 1299--1318, 2025.

\bibitem{sharma2024unpacking}
Tanusree Sharma, Yujin Potter, Kornrapat Pongmala, Henry Wang, Andrew Miller, Dawn Song, and Yang Wang.
\newblock Unpacking how decentralized autonomous organizations (daos) work in practice.
\newblock In {\em ICBC}, pages 416--424, 2024.

\bibitem{feichtinger2023hidden}
Rainer Feichtinger, Robin Fritsch, Yann Vonlanthen, and Roger Wattenhofer.
\newblock The hidden shortcomings of (d) aos--an empirical study of on-chain governance.
\newblock In {\em FC}, pages 165--185, 2023.

\bibitem{feichtinger2024sok}
Rainer Feichtinger, Robin Fritsch, Lioba Heimbach, Yann Vonlanthen, and Roger Wattenhofer.
\newblock Sok: Attacks on {DAOs}.
\newblock {\em AFT}, 2024.

\bibitem{thyagarajan2025vitarit}
Sri A.~Krishnan Thyagarajan, Easwar~Vivek Mangipudi, Lucjan Hanzlik, Aniket Kate, and Pratyay Mukherjee.
\newblock Vitarit: Paying for threshold services on bitcoin and friends.
\newblock In {\em IEEE SP}, pages 2018--2036, 2025.

\bibitem{liu2005linkable}
Joseph~K Liu and Duncan~S Wong.
\newblock Linkable ring signatures: Security models and new schemes.
\newblock In {\em ICCSA}, pages 614--623. Springer, 2005.

\bibitem{sun2017ringct}
Shi-Feng Sun, Man~Ho Au, Joseph~K Liu, and Tsz~Hon Yuen.
\newblock Ringct 2.0: A compact accumulator-based (linkable ring signature) protocol for blockchain cryptocurrency monero.
\newblock In {\em ESORICS}, pages 456--474, 2017.

\bibitem{yuen2020ringct}
Tsz~Hon Yuen, Shi-feng Sun, Joseph~K Liu, Man~Ho Au, Muhammed~F Esgin, Qingzhao Zhang, and Dawu Gu.
\newblock Ringct 3.0 for blockchain confidential transaction: Shorter size and stronger security.
\newblock In {\em FC}, pages 464--483. Springer, 2020.

\bibitem{duan2024concise}
Junke Duan, Shihui Zheng, Wei Wang, Licheng Wang, Xiaoya Hu, and Lize Gu.
\newblock Concise ringct protocol based on linkable threshold ring signature.
\newblock {\em IEEE TDSC}, 21(5):5014--5028, 2024.

\bibitem{glaeser2022foundations}
Noemi Glaeser, Matteo Maffei, Giulio Malavolta, Pedro Moreno-Sanchez, Erkan Tairi, and Sri Aravinda~Krishnan Thyagarajan.
\newblock Foundations of coin mixing services.
\newblock In {\em ACM CCS}, pages 1259--1273, 2022.

\end{thebibliography}
\end{document}